\documentclass[a4paper,preprintnumbers,floatfix,superscriptaddress,pra,twocolumn,showpacs,notitlepage,longbibliography]{revtex4-1}

\usepackage[utf8]{inputenc}
\usepackage[T1]{fontenc}
\usepackage{mathpazo}
\usepackage{amsmath, amsthm, amssymb,amsfonts,mathbbol,amstext}
\usepackage{graphicx}
\usepackage{dcolumn}
\usepackage{bm}
\usepackage{bbm}
\usepackage[usenames,dvipsnames]{xcolor}
\usepackage[]{hyperref}
\definecolor{C2}{RGB}{251, 77, 61}
\hypersetup{colorlinks=true, linkcolor=C2, citecolor=C2, urlcolor=C2}
\usepackage{mathtools}
\usepackage{comment}
\usepackage{color}
\usepackage{multirow}
\usepackage{pgf,tikz}
\usepackage{mathrsfs}
\usepackage{physics}
\usetikzlibrary{arrows}
\usepackage{xcolor}
\usepackage{soul}
\usepackage{adjustbox}
\usepackage{enumerate}
\usepackage[normalem]{ulem}
\usepackage{xcolor,cancel}

\newtheorem{theorem}{Theorem}
\newtheorem{definition}{Definition}

\newtheorem{proposition}{Proposition}
\newtheorem{lemma}{Lemma}
\newtheorem{conjecture}{Conjecture}

\newcommand\myeqA{\stackrel{\mathclap{\normalfont\mbox{$\Lambda_{1}$}}}{\longrightarrow}}
\newcommand\myeqB{\stackrel{\mathclap{\normalfont\mbox{$\Lambda_{2}$}}}{\longrightarrow}}
\newcommand\myeqC{\stackrel{\mathclap{\normalfont\mbox{$\Lambda_{3}$}}}{\longrightarrow}}
\newcommand\myeqN{\stackrel{\mathclap{\normalfont\mbox{$\Lambda_{n-1}$}}}{\longrightarrow}}
\newcommand\myeqfive{\stackrel{\mathclap{\normalfont\mbox{$\Lambda_{5}$}}}{\longrightarrow}}
\newcommand\myeqseven{\stackrel{\mathclap{\normalfont\mbox{$\Lambda_{7}$}}}{\longrightarrow}}
\begin{document}

\title{Quantum Markov monogamy inequalities}
%\date{\today}

\author{Matheus Capela}
\email{matheus@qpequi.com}
\affiliation{Institute of Physics, Federal University of Goi\'{a}s, 74.690-900, Goi\^{a}nia, Brazil}

\author{Lucas C. C\'{e}leri}
\email{lucas@qpequi.com}
\affiliation{Institute of Physics, Federal University of Goi\'{a}s, 74.690-900, Goi\^{a}nia, Brazil}

\author{Rafael Chaves}
\email{rafael.chaves@ufrn.br}
\affiliation{International Institute of Physics, Federal University of Rio Grande do Norte, 59078-970, P. O. Box 1613, Natal, Brazil}
\affiliation{School of Science and Technology, Federal University of Rio Grande do Norte, Natal, Brazil}

\author{Kavan Modi}
\email{kavan.modi@monash.edu}
%\affiliation{Shenzhen Institute for Quantum Science and Engineering and Department of Physics,
%Southern University of Science and Technology, Shenzhen 518055, China}
\affiliation{School of Physics and Astronomy, Monash University, Clayton, VIC 3800, Australia}

\date{\today}
\begin{abstract}
Markovianity lies at the heart of communication problems. This in turn makes the information-theoretic characterization of Markov processes worthwhile. Data processing inequalities are ubiquitous in this sense, assigning necessary conditions for all Markov processes. We address here the problem of the information-theoretic analysis of constraints on Markov processes in the quantum regime. We show the existence of a novel class of quantum data processing inequalities called here quantum Markov monogamy inequalities. This new class of necessary conditions on quantum Markov processes is inspired by its counterpart for classical Markov processes, and thus providing a strong link between classical and quantum constraints on Markovianity. We go on to construct a family of multitime quantum Markov monogamy inequalities, based on the process tensor formalism and that exploits multitime correlations. We then show, by means of an explicit example, that the Markov monogamy inequalities can be stronger than the usual quantum data processing inequalities. 
\end{abstract}
\maketitle

%%%%%%%%%%%%%%%%%%%%%%%%%%%%%%%%%%%%%%%%%%%%%%%%%%%%%%
%%%%%%%%%%%%%%%%%%%%%%%%%%%%%%%%%%%%%%%%%%%%%%%%%%%%%%
%%%%%%%%%%%%%%%%%%%%%%%%%%%%%%%%%%%%%%%%%%%%%%%%%%%%%%

\section{Introduction}

Markovianity plays a central role in the theory of classical information. This is enforced by describing the asymptotic encoding-decoding scheme of a memoryless channel used with no feedback by a Markovian stochastic process~\cite{shannon1948mathematical}. This leads to the principal result in information theory, the channel-coding theorem, stating the maximum rate of classical bits reliably communicated by a noisy channel to be equal to the mutual information between the channel's input and output variables maximized over all input probability distributions. Furthermore, the mutual information between variables of any Markovian stochastic process is constrained according to the so-called classical data processing inequalities (CDPIs). Indeed, a CDPI is directly used in the proof of the converse statement of the channel-coding theorem, stating that any encoding-decoding of a noisy channel with communication rate superior to its capacity is not reliable~\cite{yeung2008information}. While the classical data processing theorems are widely studied~\cite{arnold2004data, merhav2012data, kang2010new, du2017strong, kamath2015strong, raginsky2013logarithmic, merhav2011data, anantharam2014hypercontractivity}, a wider range of constraints on classical Markovian processes, the so-called monogamy inequalities, were only discovered recently~\cite{capela2020monogamy} and conjectured to hold for quantum processes as well. 

With rising interest in quantum information theory, classical data processing inequalities were extended to the regime of quantum processes~\cite{schumacher1996quantum}, setting an appropriate approach to define new constraints on quantum Markov processes. Indeed, intense research has been undertaken in this direction since then~\cite{buscemi2014complete, buscemi2016equivalence, buscemi2017comparison, buscemi2015reverse, hayden2004structure, rivas2014quantum, ferrie2014data, sai, turkmen2017data}. Naturally, the development of further constraints in terms of information inequalities imposed by quantum Markovian processes is of great interest in information theory, and it is the main problem addressed in this paper.

Here, we prove the quantum Markov monogamy inequality (QMMI) constraining any four-time-steps Markov process. The QMMI is the quantum counterpart of the classical Markov monogamy inequality (CMMI)~\cite{capela2020monogamy}, and also is the main result of this study. The QMMIs valid for six- and eight-time-steps Markov processes are also provided. This leads to a conjecture on the Markov monogamy inequalities for arbitrarily long quantum Markov processes. Furthermore, the results presented here also further enforce the connection between classical and quantum conditions on Markovianity. Similarly to the case of classical stochastic processes studied in~\cite{capela2020monogamy}, we apply the information inequalities to the problem of witnessing non-Markovianity in a quantum process. By considering a concrete example, we show that there are quantum non-Markov processes that can be witnessed by a QMMI, while not violating any quantum data processing inequalities (QDPI) from~\cite{schumacher1996quantum}. Finally, we construct a larger set of QMMIs that account for multitime quantum correlations using the process tensor (or process matrix) formalism~\cite{pollock2018non, pollock2018operational, Costa2016, Milz2020kolmogorovextension, PRXQuantum.2.030201}.

The paper is organized as follows. In Sec.~\ref{cdp} we review the classical data processing theorems, and provide the generalized version of the classical Markov monogamy inequalities. In Sec.~\ref{secqdp} we first review the quantum data processing inequality from~\cite{schumacher1996quantum}. Then we present the main result of this paper, the Markov monogamy inequality of a four-time-step quantum Markov process. Furthermore, we provide the conjecture on the general form of quantum Markov monogamy inequalities. Sec.~\ref{processtensoraproach} deals with the extension of the quantum data processing theorems, in particular the Markov monogamy inequalities, to the process tensor formalism. Finally, in Sec.~\ref{conclusion} we make our final conclusions and discuss some open problems.

\section{Classical Markov processes and information inequalities} \label{cdp}

Classically, any discrete stochastic process $\left\{X_1, \dots, X_n\right\}$ is described by a probability distribution 
\begin{gather}\label{eq:csp}
    p(x_1,\dots,x_n). 
\end{gather}
A particularly important class of processes are those called Markovian, for which the probability that the random variable $X_{i}$ takes a value $x_{i}$ at time $t_{i}$ is uniquely determined, and not affected by the possible values of $X$ at previous times to $t_{i-1}$. Mathematically, any Markovian process must have conditional probability distributions satisfying
\begin{equation}
\label{markov}
p(x_{i}\vert x_{i-1},\dots,x_{1})=p(x_{i} \vert x_{i-1}) \quad \forall i.
\end{equation}

In practice, however, it is often the case that instead of analysing the probability distribution, one is rather interested in investigating its Shannon entropy, a fundamental building block in information theory~\cite{yeung2008information}, defined for a random variable (or set of variables) $X$ as
\begin{equation}
\label{eq:shannon}
H(X) \coloneqq -\sum_{x} p(x)\log p(x),
\end{equation}
where the sum is taken over the support of $X$. Entropically, the Markov condition~\eqref{markov} is expressed as
\begin{equation}
\label{markov_ent}
H(X_i \vert X_{i-1},\dots,X_1 )= H(X_i \vert X_{i-1})\quad \forall i,
\end{equation} 
which in turn implies the paradigmatic data processing inequalities
\begin{equation}
I(X_r:X_s) \geq I(X_i:X_j) \quad {\rm with} \quad i\leq r < s \leq j.
\end{equation}
Here, $I(X_i:X_j) \coloneqq H(X_i)+H(X_j) -H(X_i:X_j)$ is the mutual information between variables $X_i$ and $X_j$.

For the simplest Markov chain with $n=3$, the only entropic constraints implied by Markovianity are the data processing inequalities given by~\cite{capela2020monogamy}
\begin{eqnarray*}
I(X_1:X_2) \geq I(X_1:X_3), \ I(X_2:X_3) \geq I(X_1:X_3),
\end{eqnarray*}
that is, we recover the usual data processing inequalities that hold for a Markov chain.

For $n\geq 4$, however, a new class of inequalities appears, generalizations of data processing called Markov monogamy inequalities, that implies constraints on the mutual information between different pairs of variables along the Markov chain~\cite{capela2020monogamy}. We present below a generalized version of the conjecture in~\cite{capela2020monogamy}, proven to hold for particular cases of $n$.

\begin{conjecture}[Classical Markov monogamy inequalities (CMMI)] \label{conj1}
Consider the Markovian process $X_n \rightarrow \cdots \rightarrow X_1 \rightarrow Y_1 \rightarrow \cdots \rightarrow Y_n$. The variables $X_1$ and $Y_1$ are to be interpreted as input and output of a given channel, respectively. The variables  $X_i$ and $Y_i$, with $i=2,\dots,n$, are interpreted as pre-processed and post-processed variables, respectively. Then for any bijective function $f \colon \{1,\dots,n\} \rightarrow \{1,\dots,n\}$, it holds that
\begin{equation}
    \sum_{i=1}^{n} I(X_i:Y_i) \geq \sum_{i=1}^{n} I(X_i:Y_{f(i)}),
\end{equation}
where $I(X:Y)$ denotes the mutual information of random variables $X$ and $Y$.
\end{conjecture}

The conjecture above has been checked for $n$ up to 4, that is, Markov chains with 8 random variables \footnote{Conditions of this type can be tested with a linear program such as ITIP \cite{yeung2016itip}}. For the case of $n=2$, we have the Markov monogamy inequality
\begin{equation}\label{M4ineq}
I(X_1:Y_1)+I(X_2:Y_2)\geq I(X_1:Y_2)+I(X_2:Y_1),
\end{equation}
associated to the bijection $f \colon \{1,2\} \rightarrow \{1,2\}$ for which $f(1)=2$ and $f(2)=1$. The remaining bijection $g \colon \{1,2\} \rightarrow \{1,2\}$ for which $g(1)=1$ and $g(2)=2$ leads to a trivial inequality, nonetheless, a valid one. Note that, independently of the number of random variables of a stochastic process, the CDPIs and CMMI are only necessary conditions for Markovianity.

It is worth noting that the result presented in Eq.~\eqref{M4ineq} appeared in Ref.~\cite{cover2005elements}, while in Ref.~\cite{capela2020monogamy}, we found an extended class of such inequalities. There we further showed that while all of the above inequalities are derived for Markov processes, they also hold for divisible processes that are non-Markovian. We also showed that there are non-Markovian processes that will also satisfy all of the above inequalities, i.e., the inequalities are necessary for any Markov process but not sufficient. Below we will generalize the Markov monogamy inequalities to the quantum case and for that we will employ the coherent information~\cite{schumacher1996quantum}. Differently from the classical case, however, our approach will not rely on the Shannon cone construction, mainly because of the non-negativity of the coherent information and the fact that the classical proof stands on the marginalization of a joint probability distribution which is not available on the quantum scenario.

\section{Quantum Markov processes and information inequalities}\label{secqdp}
Central to our purpose are the quantum Markov processes.
\begin{definition}
A sequence of quantum states $\{\rho_{1}, \cdots, \rho_{n}\}$ is a quantum Markov process with respect to the sequence of quantum channels $\{\Lambda_{1}, \cdots, \Lambda_{n-1}\}$ if the conditions $\rho_{i+1}=\Lambda_{i}(\rho_{i})$, with $i=1,\cdots,n-1$ are satisfied. This situation is denoted by
\begin{gather}\label{eq:qMarkov}
    \rho_{1} \myeqA \rho_{2} \myeqB \rho_{3}\myeqC \dots \myeqN \rho_{n}.
\end{gather}
\end{definition}

A classical channel $p(Y|X)$ transforms the classical state of the system $X$ into the state of the joint input-output system $XY$, that is, we have $p(X,Y)=p(Y|X)p(X)$. On the other hand, quantum processes are defined in a very different way. That is, a quantum channel $\Lambda$ maps the state $\rho_{1}$ of the input quantum quantum system into the state $\rho_{2}=\Lambda(\rho_{1})$ of the output quantum system. Therefore, it is not entirely trivial how one should characterize the temporal correlation in a quantum process. The development of the quantum data processing inequality provided great understanding towards this direction. Importantly, it shows we cannot directly compare classical and quantum information inequalities on equal footing. In particular, as we have stressed above, the derivation of the inequalities follow a complete different route. The classical case being based on the existence of a joint probability distribution followed a marginalization executed via a Fourier-Motzkin elimination. In turn, the quantum case combines states and channels, and uses their properties to arrive at the non-trivial quantum analogue of the information inequalities.

With the goal of deriving conditions on quantum Markov process we define the coherent information~\cite{schumacher1996quantum}. Here, we use Latin letters to denote both a quantum system and its associated Hilbert space. The coherent information of the state $\rho$ of a quantum system $\mathsf{S}_{1}$ with respect to a quantum channel $\Lambda \colon \mathsf{L}(\mathsf{S}_{1}) \rightarrow \mathsf{L}(\mathsf{S}_{2})$ is defined by
\begin{equation} \label{coherentinf}
I_{c}(\rho;\Lambda) \coloneqq H(\Lambda(\rho))-H((\mathrm{id}_{\mathsf{R}}\otimes\Lambda)(\psi)),
\end{equation}
where $\psi \in \mathsf{L}(\mathsf{R} \otimes \mathsf{S}_{1})$ is a any purification of $\rho$, and $H(\rho)$ stands for the von Neumann entropy of the quantum state, defined as
\begin{equation}
H(\rho) \coloneqq -\mathrm{Tr} \left[ \rho \log{\rho} \right ],
\end{equation}
that reduces to the usual Shannon entropy \eqref{eq:shannon} if we employ the spectral decomposition $\rho= \sum_{x}\lambda_x\ket{x}\bra{x}$ of the density operator. We will often denote the von Neumann entropy $H(\rho)$ by $H(\mathsf{S})_{\rho}$, or even by $H(\mathsf{S})$ when it is implicitly known that the quantum system $\mathsf{S}$ is in the state $\rho$. Generally, $\mathbb{1}$ denotes the identity operator, and $\mathrm{id}$ denotes the identity channel.

\subsection{Quantum data processing inequalities}

Importantly, the coherent information replaces the mutual information in the transition from CDPIs to their quantum counterparts QDPIs~\cite{schumacher1996quantum}. For any quantum state $\rho$ of $\mathsf{S}_{1}$, and for any quantum channels $\Lambda_{1} \colon \mathsf{L}(\mathsf{S}_{1}) \rightarrow \mathsf{L}(\mathsf{S}_{2})$ and $\Lambda_{2} \colon \mathsf{L}(\mathsf{S}_{2}) \rightarrow \mathsf{L}(\mathsf{S}_{3})$, it holds that
\begin{equation} \label{qdp}
I_{c}(\rho_{1};\Lambda_{1}) \geq I_{c}(\rho_{1};\Lambda_{2} \circ \Lambda_{1}).
\end{equation}
For completeness, we refer the reader to our Appendix~\ref{CohInftoMutInf} for the proof originally presented in~\cite{schumacher1996quantum}. Similarly to the interpretation in the classical case, the quantum data processing theorem states that coherent information is monotonically decreasing under the action of noisy operations.

Equivalently, instead of coherent information, it is also possible to define the quantum data processing theorem in terms of the quantum mutual information. For that, let $I(\mathsf{A}:\mathsf{B})_{\rho} \coloneqq H(\mathsf{A})+H(\mathsf{B})-H(\mathsf{A},\mathsf{B})$ denotes the mutual information of a bipartite system $\mathsf{A}\otimes\mathsf{B}$ in the state $\rho$, with $\rho_A=\mathrm{Tr}_B[\rho]$ and $\rho_B=\mathrm{Tr}_A[\rho]$. Then, for any quantum state $\rho$ of the bipartite system $\mathsf{A} \otimes \mathsf{B}$ and for any quantum channel $\Lambda \colon \mathsf{L}(\mathsf{B}) \rightarrow \mathsf{L}(\mathsf{C})$, it holds that
\begin{equation}
I(\mathsf{A}:\mathsf{B})_{\rho} \geq I(\mathsf{A}:\mathsf{C})_{(\mathrm{id}_{\mathsf{A}} \otimes \Lambda)(\rho)}.
\end{equation}
So quantum data processing theorem equivalently states that correlations between a bipartite quantum system cannot increase under the action of a local noisy operation. See Ref.~\cite{buscemi2014complete} for a presentation and application of the QDPI taking this form. For the sake of completeness, we present a proof of the equivalence of this two forms of the quantum data processing theorem in Appendix~\ref{CohInftoMutInf}.

Now we move to derive the quantum version of the four-time-steps monogamy inequality presented in Eq.~\eqref{M4ineq}.

\subsection{Quantum Markov monogamy inequality} \label{Monogamy}
In the following we prove a quantum version of the Markov monogamy inequality~\eqref{M4ineq}. This is a crucial step towards extending Conjecture~\ref{conj1}, which is valid for classical variables, also to  quantum Markov processes. As such we consider the four-time-step quantum processes of the form shown in Eq.~\eqref{eq:qMarkov}.

\begin{figure}   
\begin{adjustbox}{width=0.4\textwidth}
\begin{tikzpicture}

\draw[ultra thick,red, fill=red!10]  (0,0) rectangle (4,2);
\draw[ultra thick, CadetBlue, fill=CadetBlue!10] (3,4) rectangle (7,6);
\draw[ultra thick, CadetBlue, fill=CadetBlue!10] (6,8) rectangle (10,10);
\draw[ultra thick, CadetBlue, fill=CadetBlue!10] (9,12) rectangle (13,14);

\draw[ultra thick, cyan, fill=cyan!10]  (5.5,0) rectangle (7.5,2);
\draw[ultra thick, OliveGreen, fill=OliveGreen!10] (8.5,0) rectangle (10.5,2);
\draw[ultra thick, brown, fill=brown!10] (11.5,0) rectangle (13.5,2);

\draw[ultra thick, red]  (3.5,2) -- (3.5,4);\draw[ultra thick, cyan] (3.5,6) -- (3.5,16);
\draw[ultra thick, cyan] (6.5,2) -- (6.5,4); \draw[ultra thick, red]  (6.5,6) -- (6.5,8);\draw[ultra thick, OliveGreen] (6.5,10) -- (6.5,16);
\draw[ultra thick, OliveGreen] (9.5,8) -- (9.5,2); \draw[ultra thick, red]  (9.5,10) -- (9.5,12);
\draw[ultra thick, black] (0.5,2) -- (0.5,16);
\draw[ultra thick, brown] (12.5,12) -- (12.5,2);
\draw[ultra thick, brown] (9.5,14) -- (9.5,16);
\draw[ultra thick, red]  (12.5,14) -- (12.5,16);

\node[scale=2.5] at (2,1) {$\psi$};
\node[scale=2.5] at (6.5,1) {$\varphi_{1}$};
\node[scale=2.5] at (9.5,1) {$\varphi_{2}$};
\node[scale=2.5] at (12.5,1) {$\varphi_{3}$};
\node[scale=2.5] at (5,5) {$U_{1}$};
\node[scale=2.5] at (8,9) {$U_{2}$};
\node[scale=2.5] at (11,13) {$U_{3}$};

\draw[dashed] (0,5) -- (3,5); \draw[dashed] (7,5) -- (13.5,5);
\draw[dashed] (0,9) -- (6,9); \draw[dashed] (10,9) -- (13.5,9);
\draw[dashed] (0,13) -- (9,13); \draw[dashed] (13,13) -- (13.5,13);

\node[scale=2.5] at (1.2,3) {$\mathsf{R}$}; \node[scale=2.5] at (1.2,7) {$\mathsf{R}$};
\node[scale=2.5] at (1.2,11) {$\mathsf{R}$};
\node[scale=2.5] at (1.2,15) {$\mathsf{R}$};
\node[scale=2.5] at (4.2,3) {$\mathsf{S}_{1}$};
\node[scale=2.5] at (7.2,7) {$\mathsf{S}_{2}$};
\node[scale=2.5] at (10.2,11) {$\mathsf{S}_{3}$};
\node[scale=2.5] at (7.2,3) {$\mathsf{F}_{1}$};
\node[scale=2.5] at (4.2,6.9) {$\mathsf{E}_{1}$};
\node[scale=2.5] at (4.2,10.9) {$\mathsf{E}_{1}$};
\node[scale=2.5] at (4.2,15) {$\mathsf{E}_{1}$};
\node[scale=2.5] at (10.2,3) {$\mathsf{F}_{2}$};
\node[scale=2.5] at (10.2,6.9) {$\mathsf{F}_{2}$};
\node[scale=2.5] at (7.2,10.9) {$\mathsf{E}_{2}$};
\node[scale=2.5] at (7.2,15) {$\mathsf{E}_{2}$};
\node[scale=2.5] at (13.2,3) {$\mathsf{F}_{3}$};
\node[scale=2.5] at (13.2,6.9) {$\mathsf{F}_{3}$};
\node[scale=2.5] at (13.2,10.9) {$\mathsf{F}_{3}$};
\node[scale=2.5] at (10.2,15) {$\mathsf{E}_{3}$};
\node[scale=2.5] at (13.2,15) {$\mathsf{S}_{4}$};

\end{tikzpicture}
\end{adjustbox}
    \caption{\textbf{Diagram representing the purified process $\rho_{1} \myeqA \rho_{2} \myeqB \rho_{3} \myeqC \rho_{4}$.} The quantum state $\psi$ is a purification of $\rho_{1}$. Thus, $\rho_{1}$ is obtained from $\psi$ by tracing out the $\mathsf{R}$ system. The unitary operator $U_{i}$ and the pure state $\varphi_{i}$ provides a dilation of the quantum channel $\Lambda_{i}$, with $i=1,2,3,4.$ The remaining quantum states $\rho_{2},\rho_{3},\rho_{4}$ are obtained by acting the unitary operations and tracing out the appropriate reference-environment systems.}
    \label{Mdiagram}
\end{figure}
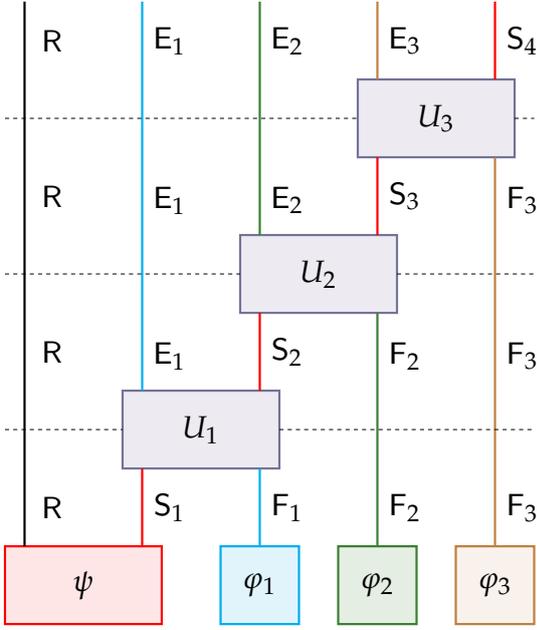

\begin{theorem}[Quantum Markov monogamy inequality (QMMI)] \label{M4thm}
For any quantum state $\rho_{1}$ of a system $\mathsf{S}_{1}$, and for any quantum channels $\Lambda_{1} \colon \mathsf{L}(\mathsf{S}_{1}) \rightarrow \mathsf{L}(\mathsf{S}_{2})$, $\Lambda_{2} \colon \mathsf{L}(\mathsf{S}_{2}) \rightarrow \mathsf{L}(\mathsf{S}_{3})$ and $\Lambda_{3} \colon \mathsf{L}(\mathsf{S}_{3}) \rightarrow \mathsf{L}(\mathsf{S}_{4})$, it holds that
\begin{gather} \label{monogamy}
\begin{split}
&I_{c}(\rho_{1},\Lambda_{3} \circ \Lambda_{2} \circ \Lambda_{1}) + I_{c}(\Lambda_{1}(\rho_{1}),\Lambda_{2}) \\  & \qquad \geq   I_{c}(\rho_{1},\Lambda_{2} \circ \Lambda_{1}) + I_{c}(\Lambda_{1}(\rho_{1}),\Lambda_{3} \circ \Lambda_{2}).
\end{split}
\end{gather}
\end{theorem}

\begin{proof}
Define a dilation of each quantum channel $\Lambda_{i}$ ---with $i=1,2,3$--- by setting a unitary operator $U_{i} \colon \mathsf{S}_{i} \otimes \mathsf{F}_{i} \rightarrow \mathsf{S}_{i+1} \otimes \mathsf{E}_{i}$ and a pure quantum state $\varphi_{i}$ for which
\begin{equation} \label{dilation3}
\Lambda_{i}(\rho)=\Tr_{\mathsf{E}_i}\left[ U_{i} (\rho  \otimes \varphi_{i})U_{i}^{\dagger} \right],   
\end{equation} 
for any operator $\rho_{1}$ in $\mathsf{L}(\mathsf{S}_{i})$. Fig.~\ref{Mdiagram} provides a representation of the process for a purification $\psi$ in $\mathsf{L}(\mathsf{R} \otimes \mathsf{S}_{1})$ of the initial quantum state $\rho_{1}$.

Now, note the definition of the coherent information terms involved in~\eqref{monogamy}, and given by
\begin{align}
I_{c}(\rho_{1};\Lambda_3 \circ \Lambda_2 \circ \Lambda_1) = H(\mathsf{S}_{4}) - H(\mathsf{R},\mathsf{S}_{4});\label{B1}\\
I_{c}(\rho_{2};\Lambda_2) = H(\mathsf{S}_{3})-H(\mathsf{R},\mathsf{E}_{1},\mathsf{S}_{3}); \label{B2} \\
 I_{c}(\rho_{1};\Lambda_2 \circ \Lambda_1) = H(\mathsf{S}_{3})-H(\mathsf{R},\mathsf{S}_{3}); \label{B3} \\
I_{c}(\rho_2;\Lambda_3 \circ \Lambda_2) =H(\mathsf{S}_{4})-H(\mathsf{R},\mathsf{E}_{1},\mathsf{S}_{4}), \label{B4}
\end{align}
with $\rho_{2}=\Lambda_{1}(\rho_{1})$.

Then, consider the following equality of the entropy terms due to the purity of the correspondent quantum systems
\begin{eqnarray}
H(\mathsf{R},\mathsf{S}_{4}) = H(\mathsf{E}_{1},\mathsf{E}_{2},\mathsf{E}_{3});\label{C1}\\
H(\mathsf{R},\mathsf{E}_{1},\mathsf{S}_{3}) = H(\mathsf{E}_2);\label{C2} \\
H(\mathsf{R},\mathsf{S}_{3})=H(\mathsf{E}_{1},\mathsf{E}_{2});\label{C3}\\
H(\mathsf{R},\mathsf{E}_{1},\mathsf{S}_{4})=H(\mathsf{E}_{2},\mathsf{E}_{3}).\label{C4}
\end{eqnarray}

Summing Eqs.~(\ref{C1},\ref{C2},\ref{C3},\ref{C4}), using the strong subadditivity of quantum entropy~\cite{lieb1973proof},
\begin{equation}
H(\mathsf{E}_{1},\mathsf{E}_{2},\mathsf{E}_{3})+H(\mathsf{E}_{2}) \leq  H(\mathsf{E}_{1},\mathsf{E}_{2})+H(\mathsf{E}_{2},\mathsf{E}_{3}),
\end{equation}
and solving for~(\ref{B1},\ref{B2},\ref{B3},\ref{B4}), we have the desired inequality.
\end{proof}

As proven in Appendix~\ref{M4toCQMI}, the quantum Markov monogamy inequality (QMMI) can equivalently be cast as the
monotonicity of the quantum conditional mutual information. This is formalized in the theorem below.

\begin{theorem}[Monotonicity of the quantum conditional mutual information~\cite{wilde2011classical}] \label{QCtheorem}
For any tripartite quantum state $\rho$ in $\mathsf{L}(\mathsf{A}\otimes \mathsf{B} \otimes \mathsf{C})$ and for any quantum channel $\Lambda \colon \mathsf{L}(\mathsf{B}) \rightarrow \mathsf{L}(\mathsf{D})$, it holds that
\begin{equation}
I(\mathsf{A}:\mathsf{B}|\mathsf{C})_{\rho} \geq I(\mathsf{A}:\mathsf{D}|\mathsf{C})_{(\mathrm{id}_{\mathsf{A}}\otimes\Lambda\otimes\mathrm{id}_{\mathsf{C}})(\rho)}.
\end{equation}
\end{theorem}

We remark that our inequalities are completely general and hold for arbitrary quantum processes (noisy or noiseless). In particular, notice that a necessary condition for non-Markovianity is the that the process is non-unitary (thus, coupled to an enviroment). In this sense, the noiseless case is not interesting from a non-Markovianity point of view.

In any case, classical data processing inequalities are directly associated with the effect of noise on the processes. For instance, the classical data processing inequality $I(X_1:X_2) \geq I(X_1:X_3)$ is satisfied with equality whenever the postprocessing stage $X_2 \rightarrow X_3$ is defined by a deterministic bijective transformation.

A similar reasoning holds for quantum information inequalities. That is, the quantum data processing inequality $I_c(\rho;\Lambda_1) \geq I_c(\rho;\Lambda_2 \circ \Lambda_1)$ is satisfied with inequality whenever the channels $\Lambda_1,\Lambda_2$ are deterministic quantum unitary operations. Moreover, the quantum Markov monogamy inequality in Eq.~\eqref{monogamy} is also satisfied with equality whenever the quantum channels $\Lambda_1,\Lambda_2,\Lambda_3$ are deterministic unitary quantum operations. Once again, this enforces non-trivial links between classical and quantum information processing conditions.

\subsection{Violation of the quantum Markov monogamy inequality} \label{QMMViolation}
Our aim here is to show that the QMMI in Eq.~\eqref{monogamy} can be violated by non-Markovian processes even in situations where all QDPI~\eqref{qdp} continue to hold. That is, we will prove that the QMMI can witness quantum non-Markovianity beyond what is possible relying solely on quantum data processing inequalities.

For a four-time-step quantum Markov process of the form of Eq.~\eqref{eq:qMarkov} the following quantities are positive semi-definite,
\begin{align} 
\mathrm{DP}_{1} \coloneqq& I_{c}(\rho_{1};\Lambda_{1})-I_{c}(\rho_{1};\Lambda_{2}\circ\Lambda_{1});\label{witness1} \\
\mathrm{DP}_{2} \coloneqq& I_{c}(\rho_{1};\Lambda_{1})-I_{c}(\rho_{1};\Lambda_{3}\circ\Lambda_{2}\circ\Lambda_{1});\label{witness2} \\
\mathrm{DP}_{3} \coloneqq& I_{c}(\rho_{1};\Lambda_{2}\circ \Lambda_{1}) - I_{c}(\rho_{1};\Lambda_{3}\circ\Lambda_{2}\circ\Lambda_{1});\label{witness3} \\
\mathrm{DP}_{4} \coloneqq& I_{c}(\rho_{2};\Lambda_{2})-I_{c}(\rho_{2};\Lambda_{3}\circ\Lambda_{2}); \label{witness4}\\
\mathrm{M4} \coloneqq& I_{c}(\rho_{1};\Lambda_{3} \circ \Lambda_{2} \circ \Lambda_{1}) + I_{c}(\rho_{2};\Lambda_{2}) \nonumber \\    &-I_{c}(\rho_{1};\Lambda_{2} \circ \Lambda_{1}) - I_{c}(\rho_{2};\Lambda_{3} \circ \Lambda_{2}), \label{witness5}
\end{align}
the first four corresponding to QDPIs of the same form as in Eq.~\eqref{qdp} and the last one corresponding QMMI in Eq.~\eqref{monogamy}. 

We do not consider the quantum version of all possible CDPIs for a four-time-step process. For instance, the quantum version of inequalities of the type $I(X_{2}:X_{3}) \geq I(X_{1}:X_{3})$ have not been considered. The reason is twofold. Firstly, the toolkit related to the proof of the QDPI ---presented in Appendix~\ref{CohInftoMutInf}--- do not directly applies to this case. Thus we leave it for future studies. Secondly, considering this type of QDPIs do not add any new information to our example. That is because inequalities of this form are not violated for the non-Markovian process examined here. Appendix~\ref{DPIextraAppendix} provides a discussion on this two claims.

To generate non-Markovian correlations (one that cannot thus be represented by the process represented in Fig.~\ref{Mdiagram}), we exploit an initially correlated tripartite system $\mathsf{R} \otimes \mathsf{S} \otimes \mathsf{E}$, with each of its parts consisting of qubit systems, and in the the pure state
\begin{equation} \label{initialstate}
\ket{\psi}=\frac{1}{\sqrt{3}}(\ket{1,0,0}+\ket{0,1,0}+\ket{0,0,1}).
\end{equation}
The collection $\{\ket{0},\ket{1}\}$ corresponds to the computational basis of the local systems.

\begin{figure}[t!]
\begin{center}

\includegraphics[width=0.5\textwidth]{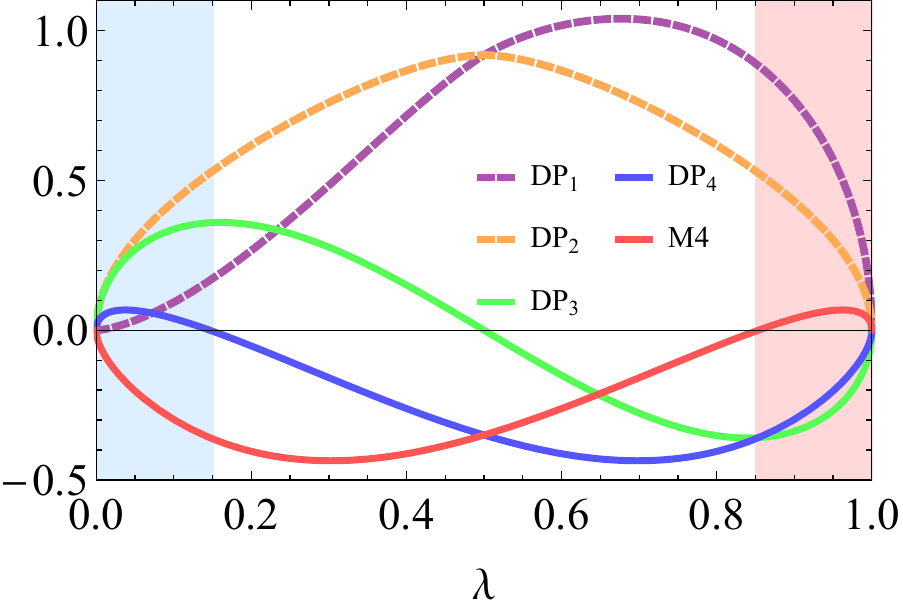}
    \caption{\textbf{Markov monogamy violation.} Markov monogamy is the only inequality being violated in the region $0 \leq \lambda \leq 0.15$ (shaded in blue). Nevertheless, the converse situation is also possible. In the region $0.85 \leq \lambda \leq 1$ (shaded in red), the monogamy inequality is not violated, while two data processing inequalities witness the non-Markovian behavior of the evolution. } \label{monogamy_violation}
\end{center}
\end{figure}

On the system-environment part of this initial state, we apply in sequence the unitary operation $U_{\lambda}$, with $0 \leq \lambda \leq 1$, that in the computational basis is given by
\begin{equation} \label{UnitaryOperation}
U_{\lambda}=\begin{bmatrix}
0 & -\sqrt{1 - \lambda} & \sqrt{\lambda} & 0 \\
1 & 0 & 0 & 0 \\
0 & 0 & 0 & 1 \\
0 & \sqrt{\lambda} & \sqrt{1 - \lambda} & 0
\end{bmatrix}.
\end{equation}
Thus, we have the well-defined sequence of states of the tripartite system $\mathsf{R} \otimes \mathsf{S} \otimes \mathsf{E}$ given by
\begin{eqnarray}
\gamma_{1}&=&\ket{\psi}\bra{\psi};\label{gamma1}\\
\gamma_{2}&=&(\mathbb{1}_{\mathsf{R}}\otimes U_{\lambda}) \gamma_{1} (\mathbb{1}_{\mathsf{R}}\otimes U_{\lambda})^{\dagger};\label{gamma2}\\
\gamma_{3}&=&(\mathbb{1}_{\mathsf{R}}\otimes U_{\lambda}) \gamma_{2} (\mathbb{1}_{\mathsf{R}}\otimes U_{\lambda})^{\dagger};\label{gamma3}\\
\gamma_{4}&=&(\mathbb{1}_{\mathsf{R}}\otimes U_{\lambda}) \gamma_{3} (\mathbb{1}_{\mathsf{R}}\otimes U_{\lambda})^{\dagger}.\label{gamma4}
\end{eqnarray}

From the above sequence of states, diagrammatically represented in Fig.~\ref{NMdiagram}, we can compute the QDPIs (denoted as $\mathrm{DP}_i$ with $i=1,\cdots,4$) and QMMI (denoted as $\mathrm{M4}$) witnessing the non-Markovianity of the local evolution of the system $\mathsf{S}$. 

We now follow Equations~(\ref{B1}-\ref{B4}), along with Equations~(\ref{witness1}-\ref{witness5}), to compute the DPIs and QMMI. For the inequality in~\eqref{witness1}, we need to compute entropies of states $\Tr_{\mathsf{R}, \mathsf{E}}[\gamma_{i}]$ and $\Tr_{ \mathsf{E}}[\gamma_{i}]$ for $i=2,3$.
The entropies of the first terms, in Equations~(\ref{B1}-\ref{B4}), are simply the entropies of states $\Tr_{\mathsf{R}, \mathsf{E}}[\gamma_{3}]$ and $\Tr_{\mathsf{R}, \mathsf{E}}[\gamma_{4}]$.  The entropies of the second terms in Equations~(\ref{B1},\ref{B3}) are with respect to $\Tr_{\mathsf{R}}[\gamma_{4}]$ and $\Tr_{\mathsf{R}}[\gamma_{3}]$, respectively. The fact that $\mathsf{R}$ is not a purification of the initial state of $\mathsf{S}$, will play a crucial role in violating the inequalities. In other words, the initial correlations with $\mathsf{E}$ are a non-Markovian feature. On the other hand, the entropies of the second terms in Equations~(\ref{B2},\ref{B4}) are with respect to the total states $\gamma_{3}$ and $\gamma_{4}$, respectively. This is because, unlike in a Markov process, $\mathsf{E}_1$ is the same as $\mathsf{E}$, i.e., the total environment. This leads coherent interference of $\mathsf{SE}$ correlations throughout the process, which too will play a central role in violating the inequalities below. This leads to 
\begin{align} 
\mathrm{DP}_{1} \coloneqq& [H(\mathsf{S})-H(\mathsf{R},\mathsf{S})]_{\gamma_{2}}-[H(\mathsf{S})-H(\mathsf{R},\mathsf{S})]_{\gamma_{3}};\label{witness1B} \\
\mathrm{DP}_{2} \coloneqq& [H(\mathsf{S})-H(\mathsf{R},\mathsf{S})]_{\gamma_{2}}-[H(\mathsf{S})-H(\mathsf{R},\mathsf{S})]_{\gamma_{4}};\label{witness2B} \\
\mathrm{DP}_{3} \coloneqq& [H(\mathsf{S})-H(\mathsf{R},\mathsf{S})]_{\gamma_{3}} - [H(\mathsf{S})-H(\mathsf{R},\mathsf{S})]_{\gamma_{4}};\label{witness3B} \\
\mathrm{DP}_{4} \coloneqq& [H(\mathsf{S})-H(\mathsf{R},\mathsf{S},\mathsf{E})]_{\gamma_{3}} \nonumber \\& - [H(\mathsf{S})-H(\mathsf{R},\mathsf{S},\mathsf{E})]_{\gamma_{4}}; \label{witness4B}\\
\mathrm{M4} \coloneqq& [H(\mathsf{R},\mathsf{S},\mathsf{E})-H(\mathsf{R},\mathsf{S})]_{\gamma_{4}} \nonumber \\
&- [H(\mathsf{R},\mathsf{S})-H(\mathsf{R},\mathsf{S},\mathsf{E})]_{\gamma_{3}}. \label{witness5B}
\end{align}

The quantities above are presented in Fig.~\ref{monogamy_violation} for the process represented by Equations~(\ref{gamma1},\ref{gamma2},\ref{gamma3},\ref{gamma4}). We show there are processes for which the Markov monogamy inequality is violated, and thus witnessing non-Markovianity, while none of the QDPIs being efficient in this task. We notice that the converse behavior is also possible.

\begin{figure}   
\begin{adjustbox}{width=0.4\textwidth}
\begin{tikzpicture}[rotate=-90]

\draw[ultra thick,brown, fill=brown!10]  (0,0) rectangle (7,2);
\draw[ultra thick, CadetBlue, fill=CadetBlue!10] (3,4) rectangle (7,6);
\draw[ultra thick, CadetBlue, fill=CadetBlue!10] (3,8) rectangle (7,10);
\draw[ultra thick, CadetBlue, fill=CadetBlue!10] (3,12) rectangle (7,14);

\draw[ultra thick, red]  (3.5,2) -- (3.5,4);
\draw[ultra thick, red] (3.5,6) -- (3.5,8);
\draw[ultra thick, red] (3.5,10) -- (3.5,12);
\draw[ultra thick, red] (3.5,14) -- (3.5,16);
\draw[ultra thick, cyan] (6.5,2) -- (6.5,4); \draw[ultra thick, cyan]  (6.5,6) -- (6.5,8);\draw[ultra thick, cyan] (6.5,10) -- (6.5,12);\draw[ultra thick, cyan] (6.5,14) -- (6.5,16);
\draw[ultra thick, black] (0.5,2) -- (0.5,16);

\node[scale=2.5] at (3.5,1) {$\psi$};
\node[scale=2.5] at (5,5) {$U_{1}$};
\node[scale=2.5] at (5,9) {$U_{2}$};
\node[scale=2.5] at (5,13) {$U_{3}$};

\draw[dashed] (0,5) -- (3,5); \draw[dashed] (7,5) -- (7.5,5);
\draw[dashed] (0,9) -- (3,9); \draw[dashed] (7,9) -- (7.5,9);
\draw[dashed] (0,13) -- (3,13); \draw[dashed] (7,13) -- (7.5,13);

\node[scale=2.5] at (1.2,3) {$\mathsf{R}$}; \node[scale=2.5] at (1.2,7) {$\mathsf{R}$};
\node[scale=2.5] at (1.2,11) {$\mathsf{R}$};
\node[scale=2.5] at (1.2,15) {$\mathsf{R}$};
\node[scale=2.5] at (4.2,3) {$\mathsf{S}$};
\node[scale=2.5] at (7.2,7) {$\mathsf{E}$};
\node[scale=2.5] at (7.2,3) {$\mathsf{E}$};
\node[scale=2.5] at (4.2,6.9) {$\mathsf{S}$};
\node[scale=2.5] at (4.2,10.9) {$\mathsf{S}$};
\node[scale=2.5] at (4.2,15) {$\mathsf{S}$};
\node[scale=2.5] at (7.2,10.9) {$\mathsf{E}$};
\node[scale=2.5] at (7.2,15) {$\mathsf{E}$};

\end{tikzpicture}
\end{adjustbox}
    \caption{\textbf{Non-Markov process.} The non-Markovian behavior considered here consists of initial system-environment correlations, and environmental quantum memory through the unitary system-environment evolution.}
      \label{NMdiagram}
\end{figure}
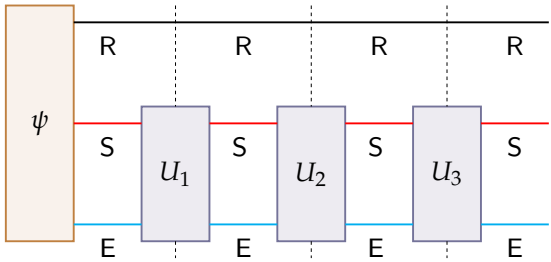

\subsection{A conjecture on quantum Markov processes}\label{conjecture}
At this stage we are ready to extend Conjecture~\ref{conj1} to the quantum realm. We take the following notation to simplify its statement. For any quantum Markov process 
\begin{equation*}
    \rho_1 \myeqA \rho_2 \myeqB \cdots \myeqN \rho_n 
\end{equation*}
we define $ I_{c}(\rho_{r}:\rho_{s})= I_{c}(\rho_{s}:\rho_{r})$, with
\begin{equation}
    I_{c}(\rho_{r}:\rho_{s}) \coloneqq I_{c}(\rho_{r};\bigcirc_{i=r}^{s-1} \Lambda_{i}),
\end{equation}
for $r<s$, where $\bigcirc_{i=r}^{s-1} \Lambda_{i} \coloneqq \Lambda_{s-1} \circ \cdots \circ \Lambda_{r}$.

\begin{conjecture}[Quantum Markov monogamy inequalities] \label{conj2}
For any quantum Markov process
\begin{equation*}
    \rho_n \rightarrow \cdots \rightarrow \rho_1 \rightarrow \sigma_1 \rightarrow \cdots \rightarrow \sigma_n,
\end{equation*}
and for any bijective function $f:\{1,\cdots,n\}\rightarrow\{1,\cdots,n\}$, it holds that
\begin{equation}
   \sum_{i=1}^{n} I_{c}(\rho_{i}:\sigma_{i}) \geq \sum_{i=1}^{n} I_{c}(\rho_{i}:\sigma_{f(i)}).
\end{equation}
\end{conjecture}

In Subsection~\ref{Monogamy} we have shown the case $n=2$ to be true. In Appendices~\ref{M6} and \ref{M8} we show that the validity of Conjecture~\ref{conj2} holds for $n=3,4$ as well.

\section{Comparison to previous results}

The witnessing of classical and quantum non-Markovianity is a key application of the classical and quantum information inequalities presented in the manuscript. Therefore, this perspective further enforces the relevance and applicability of the novel quantum Markov monogamy inequalities developed. In fact, there are many indicators for non-Markovian quantum phenomenon as discussed in Refs.~\cite{rivas2014quantum,Li2018,PRXQuantum.2.030201}. However, all of these indicators, in one manner or another, look for departures from the divisibility of the process, see~Ref.~\cite{PRXQuantum.2.030201} for details. This includes the indicators in Refs.~\cite{breuer2016colloquium,rivas2014quantum}, as well as the quantum data processing inequalities (QDPIs).

However, for a concrete comparison we note that it is well-known that the non-Markovian indicator by Breuer et al. \cite{breuer2016colloquium} is strictly weaker than that due to Rivas et al. \cite{rivas2014quantum}. However, our QDPIs already include the non-Markovian indicator of Rivas et al. To see this, we can rewrite Eq.~(11) as
\begin{equation}
    I_c(\rho;\Lambda_t) \ge I_c(\rho; \Lambda_{t+dt}).
\end{equation} \label{eq:qdpi_time}
This inequality may be violated for some non-Markovian processes, i.e.,
\begin{multline}
    I_c(\rho; \Lambda_{t+dt}) - I_c(\rho;\Lambda_t) > 0 \quad  \\ \Rightarrow \quad  \mbox{non-Markovianity.}
\end{multline}
Dividing the above equation by $dt$ tells us that when the derivative of the coherent information $\partial_t{I_c}(\rho;\Lambda_t)$ is positive we have non-Markovianity.
Thus, we can define a measure for non-Markovianity as
\begin{equation}
    \mathcal{N}_c := \max_{\rho} \int_0^T n_c(t) \ dt
\end{equation}
where the maximization is over all initial input states $\rho$, and
\begin{equation}
    n_c(t) = \max\{0, \partial_t{I_c}(\rho;\Lambda_t)\}.
\end{equation}
This is precisely the definition of non-Markovianity by Rivas et al. As there may be subtle differences in the choice of metric, a detailed comparison requires a careful study, which is beyond the scope of the present paper. See Ref.~\cite{chruscinski2011measures} for a detailed comparison between the works of Bruer et al. and Rivas et al.

The above results are obtained from only one of our QDPIs. Thus, the QMMIs may be able to see non-Markovianity where all indicators that are reliant on divisibility are blind.

We note that a similar reasoning can be applied to the QMMI in order to define a new measure of non-Markovianity. Similarly, our QMMI implies
\begin{multline}
 [I_c(\rho; \Lambda_{t+dt} \circ \Lambda_s) - I_c(\rho;\Lambda_t \circ \Lambda_s) ] \\ - [I_c(\Lambda_s(\rho); \Lambda_{t+dt} ) - I_c(\Lambda_s(\rho);\Lambda_t ) ]  > 0 \quad  \\ \Rightarrow \quad  \mbox{non-Markovianity for } t>s.
\end{multline}
Therefore, we can define the new measure for non-Markovianity
\begin{equation}
    \mathcal{M}_{c} := \max_{\rho} \int_{t>s} \int_{0}^T  m_c(t) ds dt,
\end{equation}
with 
\begin{multline}
    m_c(s,t) \\ := \max \{0,\partial_t [ I_c(\rho;\Lambda_t \circ \Lambda_s) - I_c(\Lambda_s(\rho);\Lambda_t ] \}.
\end{multline}

The conjecture on QMMIs in Sec.~\ref{conjecture} suggests how we can build novel non-Markovianity  measures analogously as above. In fact, there is a necessary and sufficient condition on divisible processes \cite{buscemi2016equivalence}.

The understanding on the limits on processing information in communication systems is of primal concern in information theory. In this respect, data processing inequalities have been shown to be of fundamental relevance in the development of main results in classical information theory \cite{cover2005elements}. Precisely, the data processing inequality $I(X_1:X_4) \leq I(X_2:X_3)$ held by four-time-step classical processes, is the mathematical result supporting the derivation of the converse part of the channel-coding theorem: there is no reliable asymptotic encoding-decoding scheme with communication rate larger than the channel capacity. Clearly, the channel-coding theorem stems as a fundamental result in information science, and therefore, makes sure the relevance on the development of further information inequalities. In turn, the Markov monogamy inequalities appears as constraints on information processing complementary to data processing inequalities. 

On the other hand, quantum processes differ fundamentally from their classical counterpart, and thus, demand further understanding and analysis. In this sense, the quantum data processing inequality is regarded as a highly non-trivial result, and also as one of the pillars in quantum information theory \cite{nielsen2002quantum}. Therefore, it is not clear beforehand that a given classical information inequality has a quantum counterpart. The same is true for the Markov monogamy inequalities. The techniques we had to employ to prove their quantum analogue are completely different from what is used in the classical case. 

To see why this is the case, we notice that in the classical case the proof of inequalities rely on the existence of a joint probability distribution $p(x_1,\dots,x_n)$ that marginalizes (via a quantifier elimination implemented by the Fourier-Motzkin algorithm) to the pairwise distribution $p(x_i,x_j)$ used in the definition of the mutual information. In turn, in the quantum case, since we are dealing with quantum states at different time steps, there is a priori no joint description. Instead, we have to rely on a ''mixed'' description in terms of channels and states, precisely the reason why we employ the coherent information.
Thus, even though our work is motivated by the classical monogamy inequality, it is not a trivial or natural extension of it and thus, develops a fundamental link between classical and quantum information inequalities. 

Furthermore, the quantum Markov monogamy inequalities display novel constraints on the processing of quantum systems, being capable of witnessing non-Markovianity in a regime where the paradigmatic data processing inequalities (widely used in the literature) would simply fail to do so.

Now, we consider multitime correlations indicators based on the process tensor, which are able to also see non-Markovian features in divisible processes~\cite{milz-kim}. In what follows, we show how our quantum Markov monogamy inequalities (QMMIs) are written in terms of multitime correlations and thus are stronger indicators for quantum non-Markovian phenomena than the ones due to Breuer et al. \cite{breuer2016colloquium} and Rivas et al. \cite{rivas2014quantum}.

\section{Quantum stochastic processes} 
\label{processtensoraproach}

The CDPI and CMMI stem from a well-defined notion of stochastic process, namely Eq.~\eqref{eq:csp}. The quantum inequalities, in contrast, are derived for family of quantum channels. This raises the question if there is a quantum equivalent of Eq.~\eqref{eq:csp}? If so, can we derive a larger family of inequalities than the ones given in the last section?

In this Section, we will work with the process tensor framework, which is a natural generalisation of Eq.~\eqref{eq:csp} for quantum processes. With this we will derive another family of QDPIs and QMMIs. Importantly, the set of inequalities in the previous section will be satisfied by divisible processes~\cite{rivas2014quantum}, even when the process is non-Markovian. This is because they only account for two-time correlations, and neglect higher-order correlations in the process~\cite{milz-kim}. In contrast, the forthcoming family of inequalities account for multitime correlations and will be capable of identifying the non-Markovian features in such processes. We begin by first reviewing the fundamental elements of this framework.

\subsection{Process Tensor}

We now discuss the structure of multitime correlations in the quantum case by considering an initial reference-system-environment state $\psi$. Before any dynamical evolution, an intervention with a control operation $\mathrm{A}_{1}$ can be made on the system alone: $\mathrm{A}_{1} \colon \mathsf{L}(\mathsf{S}_{1}) \rightarrow \mathsf{L}(\mathsf{S}_{1}')$. Next, as before, the system-environment state undergoes an evolution $\mathrm{U}_{1}:\mathsf{L}(\mathsf{S}_1'\otimes\mathsf{F}_1) \to \mathsf{L}(\mathsf{S}_2\otimes\mathsf{E}_2)$ and an intervention $\mathrm{A}_{2}$ is then made on the system $\mathsf{S}$ alone. The process repeats and the total state once again evolves due to $\mathrm{U}_{2}$, followed by a third intervention $\mathrm{A}_{3}$ on $\mathsf{S}$ alone, and so on up to a final intervention $\mathrm{A}_{4}$ is performed following $\mathrm{U}_{3}$, see Fig.~\ref{fig:processtensor}(a). Here, $\mathsf{S}_i$ is isomorphic to $\mathsf{S}_i'$ for any interventional time-step $i=1,2,3$. 

The interventions $\{\mathrm{A}_{j}\}$ are any physically implementable operation, which can be thought of a generalised measurement with possible corresponding outcomes $\{x_j\}$. Mathematically, these are known as \textit{instruments}~\cite{davies} and represented by a collection of completely positive  maps $\mathcal{J} := \{\mathrm{A}_{x_j}\}$ such that $\sum_{x_j} \mathrm{A}_{x_j}$ is trace preserving.

\begin{figure}[t!]
\begin{center}
\includegraphics[width=0.49\textwidth] {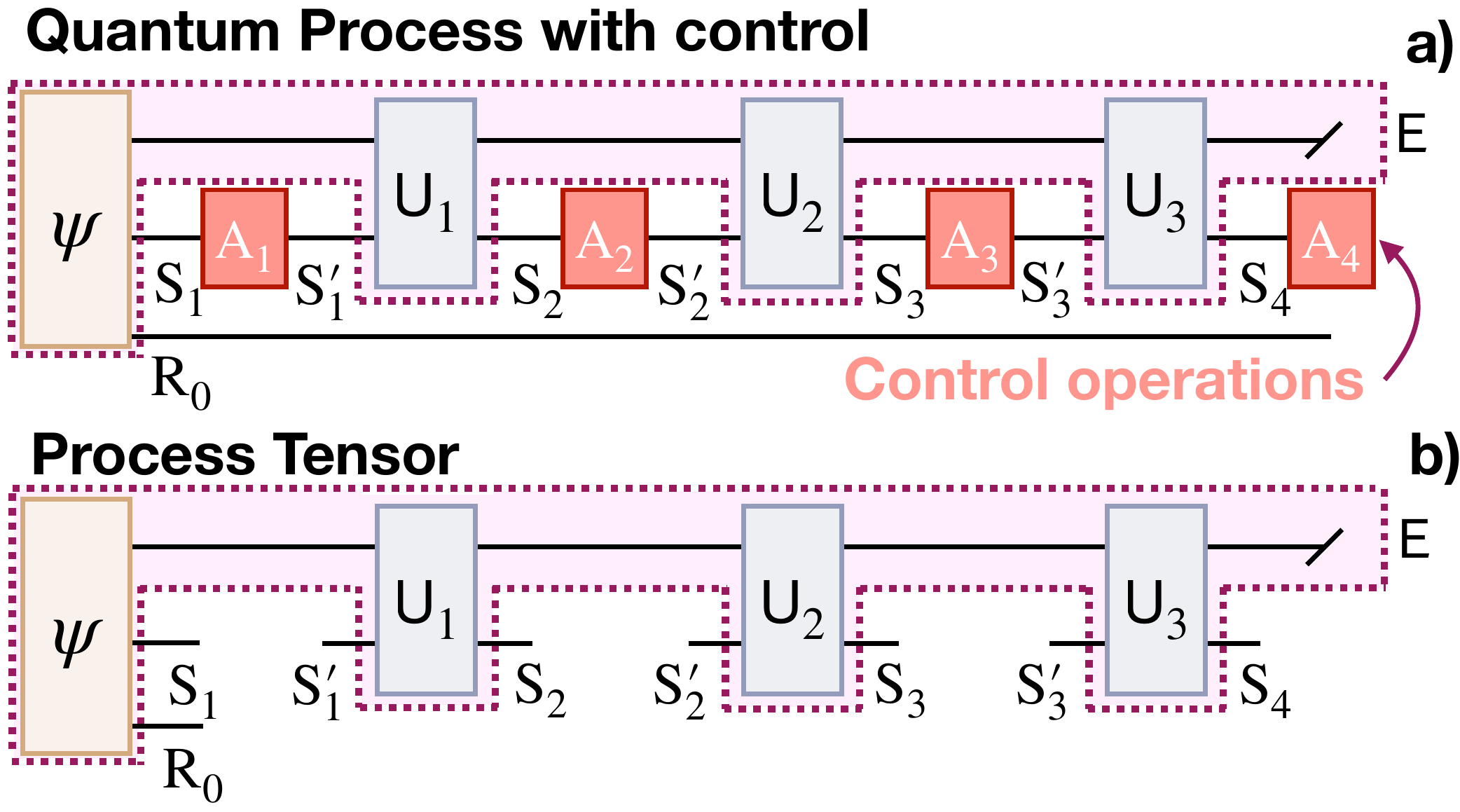}
    \caption{\textbf{Process Tensor.} The top panel shows a quantum circuit that accounts for multitime correlations for observables $\mathrm{A}_1,\cdots,\mathrm{A}_4$, i.e. Eq.~\eqref{eq:qtocl}. The object in the dotted line is called the process tensor, defined in Eq.~\eqref{eq:proctens}. It is drawn in isolation in the bottom panel, uniquely characterised the process.} \label{fig:processtensor}
\end{center}
\end{figure}

The above machinery straightforwardly allows for the calculation of the probability to observe a sequence of quantum events $(x_k, \dots x_1)$, corresponding to a choice of instruments $\{\mathcal{J}_k, \dots, \mathcal{J}_1\}$, as
\begin{gather}
\begin{split}\label{eq:qtocl}
&p(x_k,\dots,x_1\,|\,\mathcal{J}_k,\dots, \mathcal{J}_1)=\\ 
& \qquad\mbox{tr}[\mathrm{A}_{x_k} \mathrm{U}_{k-1} \!\cdots \mathrm{U}_{1} \mathrm{A}_{x_1} (\psi)].
\end{split}
\end{gather}
This is depicted in Fig.~\ref{fig:processtensor}(a). Here, the LHS is akin to a classical joint probability distribution. We can identify the quantum stochastic process by rewriting the RHS as
\begin{align}
&\mbox{tr}[\mathrm{A}_{x_k} \mathrm{U}_{k-1} \!\cdots \mathrm{U}_{1} \mathrm{A}_{x_1} (\psi)] = \mbox{tr}[\Upsilon_{k:1} \mathbf{A}^{\rm T}_{k:1}] \label{eq:process}\\ \label{eq:control}
\mbox{with} \quad &\mathbf{A}_{k:1} \!:=\! \mathrm{A}_{x_k} \!\otimes\! \dots \otimes \mathrm{A}_{x_1}\\ \label{eq:proctens}
\mbox{and} \quad &\Upsilon_{k:1} \!:=\! \mbox{tr}_B [\mathrm{U}_k \star \dots  \mathrm{U}_1 \star \psi]
\end{align}
where ${\rm T}$ denotes transposition and $\star$ denotes the link product, defined as a matrix product on the space $E$ and a tensor product on space $S$~\cite{PhysRevA.80.022339}. Here, $\mathrm{U}$ and $\mathrm{A}$ are the Choi operators of the corresponding transformations. The important feature here is the clear separation of the interventions $\mathbf{A}_{k:1}$ from the influences due to the bath, which are packaged in the \textit{process tensor} $\Upsilon_{k:1}$~\cite{pollock2018non, pollock2018operational, Costa2016}.

The process tensor is depicted inside the red-dotted line in Fig.~\ref{fig:processtensor}(b), and usually denoted by its \emph{Choi state} $\Upsilon_{k:1}$~\cite{PRXQuantum.2.030201}. It is the quantum generalisation of the joint classical probability distribution and unambiguously represents a quantum stochastic process~\cite{Milz2020kolmogorovextension}, and reduces to the classical case in the right limits~\cite{PhysRevA.100.022120, arXiv:1907.05807}. It contains all accessible multitime correlations~\cite{PhysRevLett.114.090402, White2020, White2021a}.

\begin{figure}[t!]
\begin{center}
\includegraphics[width=0.49\textwidth] {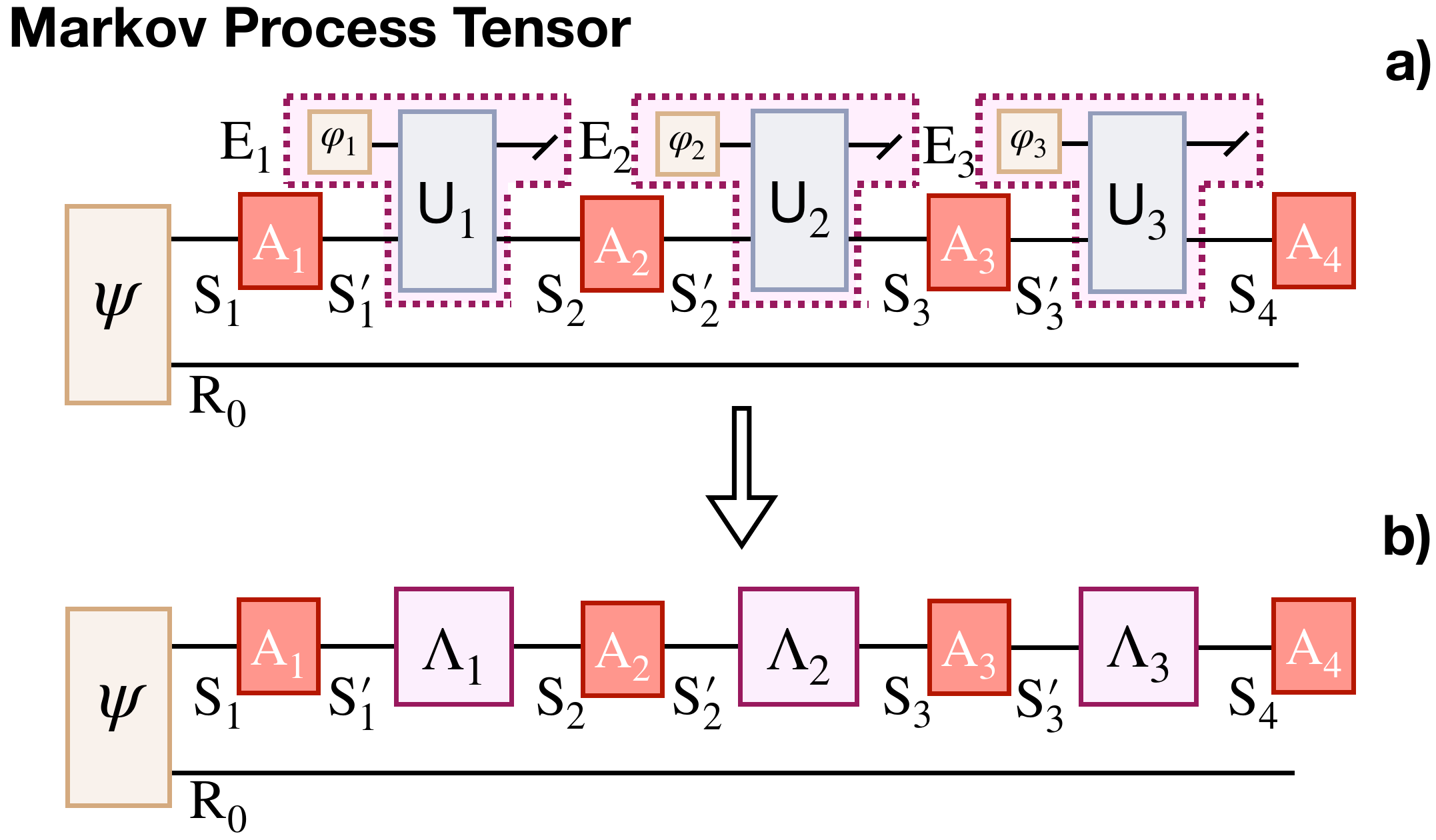}
    \caption{\textbf{Markov process tensor.} The top panel shows a quantum circuit for a Markovian process, in presence of interventions. This clearly reduces to what we posit a Markov process in the previosu section.} \label{fig:processMarkov}
\end{center}
\end{figure}

An important result that stems from the process tensor formalism is a necessary and sufficient condition for quantum Markov processes~\cite{pollock2018operational, Costa2016}. Namely, Markovian processes are those satisfying the following property: any $k$-time process tensor factorises as
\begin{gather} \label{eq:markovchoi}
\Upsilon_{k:1} = \mathrm{L}_{k-1} \otimes \cdots \otimes \mathrm{L}_{2} \otimes  \mathrm{L}_{1} \otimes \rho,
\end{gather}
where $\mathrm{L}_j$ are the Choi operators of the CPTP maps $\Lambda_j$ as before. See Fig.~\ref{fig:processMarkov} for a graphical depiction. 

It is possible to show that when a quantum process is Markovian, and the interventions are rank-one projective measurements, then the resulting distribution in Eq.~\eqref{eq:qtocl} will be a Markovian distribution~\footnote{A coarse (higher-rank projection) measurement can hide memory in the degenerate subspace, even if the process is Markovian, see~\cite{taranto2019structure}.}. We readily get an infinite family of CDPI and CMMI. We, of course, also get the QDPIs and QMMI from the last section. 

The above Markov condition means that we can deduce a process to be non-Markovian by looking for correlations. In the last section, the QDPIs and QMMI are constructed by considering two-time correlations; but in general, we can certainly look at higher-order correlations. With this in mind, we aim to find a family of more general inequalities by exploiting the structure of quantum stochastic processes. In particular, we will show that the inequalities of the last section cannot differentiate divisible processes from Markovian ones. The forthcoming inequalities will not have this limitation.

\subsection{Choi state DPIs}
\label{idpi}
We first consider a family of QDPIs for a four step quantum process as before. However, we now label the space of the process tensor as $(\mathsf{R}_0, \mathsf{S}_1, \mathsf{R}_1; \mathsf{S}_2, \mathsf{R}_2; \mathsf{S}_3, \mathsf{R}_3; \mathsf{S}_4)$, see Fig.~\ref{fig:processtensor} for an illustration. The advantage we now have is that we can intervene in a physically allowed form. For instance, at a given time we may throw away the system, say $\mathsf{S}_1$, and replace it with a new system with its own reference $\mathsf{R}_1$. This allows us to construct DPIs on the future dynamics without worrying about the output states of the past dynamics.

This is clearly desirable in several situations. For instance, consider a process where $\Lambda_1$ is a fully depolarising process, but the subsequent process is rich in structure. In such cases, why should we limit ourselves to only the output of $\Lambda_1$? In such instances we would simply swap out the output of $\Lambda_1$ with a more coherent state.

We now construct a family of QDPIs based on the Choi states of a Markov process. To do so, first note that the mutual information of a bipartite state $\rho$ can be expressed as a quantum relative entropy
\begin{gather}\begin{split}
    I(\mathsf{A}:\mathsf{B})_{\rho_{AB}} &= S[\rho_{AB}\|\rho_{A}\otimes \rho_{B}]\\
    &= \Tr[\rho_{AB}\{\log(\rho_{AB}) - \log(\rho_{A}\otimes \rho_{B})\}].
\end{split}    
\end{gather}
Next, note that quantum relative entropy is contractive under the action of CPTP maps. We can thus derive the following QDPIs when the process is Markov: 
\begin{align}
& I(\mathsf{R}_1:\mathsf{S}_2) \ge I(\mathsf{R}_1:\mathsf{S}_3) \ge I(\mathsf{R}_1:\mathsf{S}_4) \\
&I(\mathsf{R}_2:\mathsf{S}_3) \ge I(\mathsf{R}_2:\mathsf{S}_4).
\end{align}
Here, we have allowed for interventions $\mathrm{A}_k$ at the intermediate times.

The mutual information above is defined 
on any input state. We now restrict ourselves to inserting one half of a maximally entangled state in the first port and computing the quantum mutual information between the other half and what comes out at the final port. In other words, we are restricting ourselves to the Choi states of the process. Notice that this reduces to the usual classical case when we dephase in a local bases and the intermediate operators $\mathrm{A}_k$ become the classical identity channels.

This allows us to derive another set of QDPIs by using the fact that the quantum relative entropy is contractive under CP maps~\cite{petz}. Consider the process from $\mathsf{R}_1$ to $\mathsf{S}_3$ and compare that to the process from $\mathsf{R}_2$ to $\mathsf{S}_3$. The Choi state for the former is $\Lambda_{2} \circ \mathrm{A}_{2} \circ \Lambda_{1} (\Psi^+)$. Graphically we can represent this as top panel in Fig.~\ref{fig:processDPI}, which we can transform into the bottom panel by simply sliding the boxes, where we have added an identity channel $\mathsf{id}_1$. Let us compare that to the process in the third panel, which indeed the Choi state for the process from $\mathsf{R}_2$ to $\mathsf{S}_3$. Since the two processes differ by actions of CP maps on the the bottom leg we have 
\begin{gather}
I(\mathsf{R}_2:\mathsf{S}_3) \ge I(\mathsf{R}_1:\mathsf{S}_3).
\end{gather}
We have given a graphical proof in Fig.~\ref{fig:processDPI}, and also see App.~\ref{cpslide}. By same argument we can also derive
\begin{gather}
I(\mathsf{R}_3:\mathsf{S}_4) \ge I(\mathsf{R}_2:\mathsf{S}_4) \ge I(\mathsf{R}_1:\mathsf{S}_4).
\end{gather}
By choosing the $\mathrm{A}$'s to be rank-one projections we recover classical processes. For all such processes we can derive the Markov monogamy condition $\mathrm{M4}$.

Note that for the quantum case we have twice the number of legs than for classical case. In the classical case we just have $\{1,2,3,4\}$, which means that we have total of six mutual informations and fifteen pairwise relations between these mutual information. Most of these are not independent, thus we wind up with four QDPIs and QMMI. In the quantum case we have $\{\mathsf{S}_1, \mathsf{R}_1; \mathsf{S}_2, \mathsf{R}_2; \mathsf{S}_3, \mathsf{R}_3; \mathsf{S}_4\}$. Therefore we have 21 mutual informations. Of these six are vanishing because of causality, i.e., $I(\mathsf{R}_y:\mathsf{S}_x)=0$ for all $y \ge x$~\cite{PRXQuantum.2.030201}. For a Markov process we can require the $\mathsf{R}$ spaces to be independent, leading to three more vanishing constraints, $I(\mathsf{R}_x:\mathsf{R}_y)=0$. This requirement also means $I(\mathsf{S}_x:\mathsf{S}_y)=0$, that means six more vanishing mutual information. We are then left with exactly the same six non-trivial mutual informations as in the classical case.

\begin{figure}[t!]
\begin{center}
\includegraphics[width=0.45\textwidth] {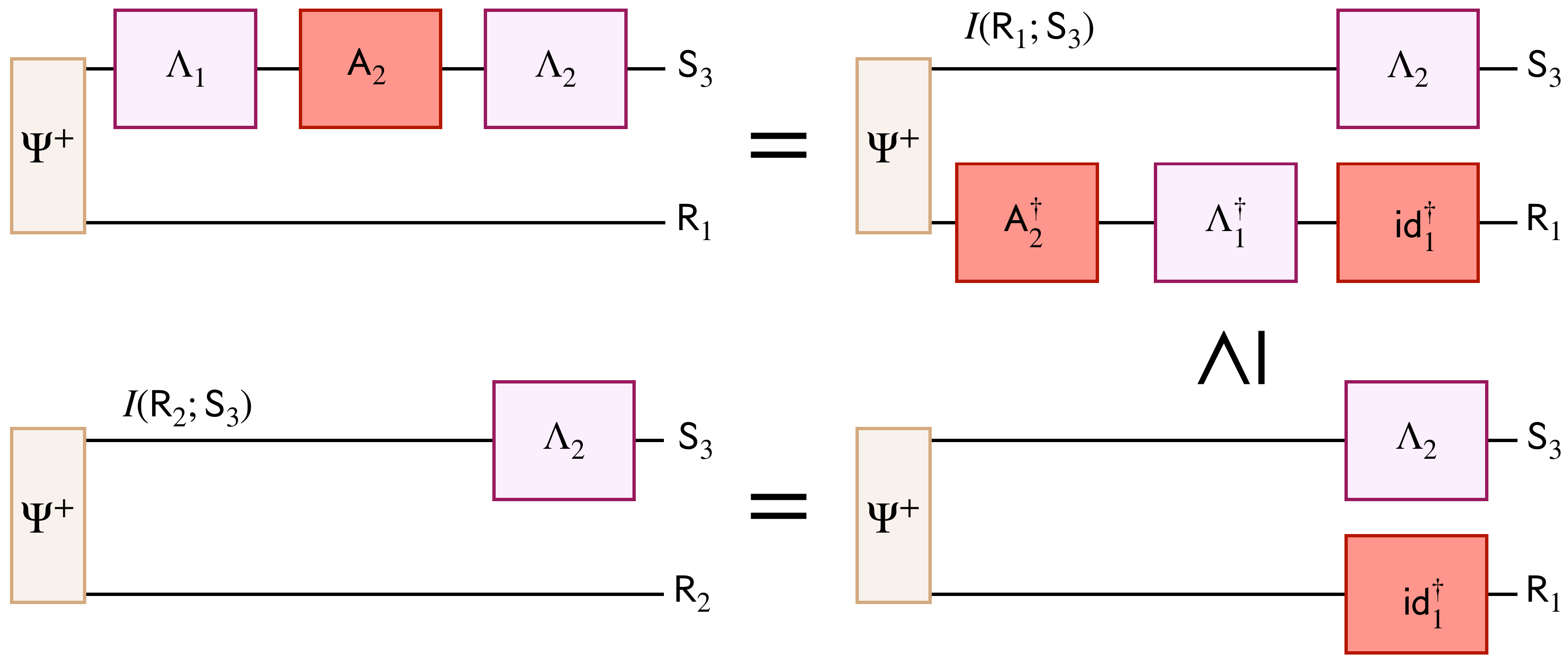}
    \caption{\textbf{Choi state data processing inequalities.} We can derive an anaologue of DPIs from the previous section by using contractivity of quantum relative entropy and some identities of actions of CPTP maps on the maximally entangled state $\Psi^+$.} \label{fig:processDPI}
\end{center}
\end{figure}

\subsection{Multitime quantum Markov monogamy inequalities}
\label{im4}

Here, we will present several families of quantum Markov monogamy inequalities. Their importance is highlighted by the fact that the QDPIs and the QMMI presented in the last section would all be satisfied for divisible non-Markovian processes. Here, we allow for interventions onto the process, which enables the detection temporal correlations that lie in divisible processes, see Ref.~\cite{milz-kim} for explicit examples. To construct the family of QMMI we will follow the circuits in Fig.~\ref{fig:processM4}.

The first family will be defined in terms of the following definition of coherent information
\begin{gather}\label{eq:qchi-v1}
    I_{q_{1}}(j;k) := H(\mathsf{S}_j,\mathsf{R}_j) - H(\mathsf{S}_j,\mathsf{R}_j,\mathsf{S}_k).
\end{gather}
Again, the space relevant for entropies in the last equation are labelled in Fig.~\ref{fig:processM4}.

\begin{theorem}[Multitime quantum Markov monogamy in equality (MQMMI-1)]
\label{iqM4thm}
For any Markov process tensor $\Upsilon_{4:1}$, it holds that
\begin{gather} \label{iqmonogamy}
I_{q_{1}}(1;4) + I_{q_{1}}(2;3) \geq   I_{q_{1}}(1;3) + I_{q_{1}}(2;4),
\end{gather}
with the intervention $\mathrm{A_j}$ in $I_{q_{1}}(j;k)$ defined by the purification of the system $\mathsf{S}_{j}$.
\end{theorem}

\begin{proof}
The above theorem considers a setup that allows for interventions $\mathrm{A}_{j}$ along the way. For simplicity, let us first consider the interventions to be identity operations.

For Markov process we have $ H(\mathsf{S}_j,\mathsf{R}_j) = H(\mathsf{S}_j)+H(\mathsf{R}_j)$ and 
$H(\mathsf{S}_j,\mathsf{R}_j,\mathsf{S}_k)=H(\mathsf{S}_j)+H(\mathsf{R}_j,\mathsf{S}_k)$. Next, we define a dilation of each quantum channel $\Lambda_{i}$ as done in Eq.~\eqref{dilation3}. Now, note the definition of $I_q$ in in Eq.~\eqref{eq:qchi-v1} becomes
\begin{align}
I_{q_{1}}(1;4) = H(\mathsf{R}_1) - H(\mathsf{E}_1,\mathsf{E}_2,\mathsf{E}_3); \label{iqB1}\\
I_{q_{1}}(2;3) = H(\mathsf{R}_2) - H(\mathsf{E}_2); \label{iqB2} \\
I_{q_{1}}(1;3) = H(\mathsf{R}_1) - H(\mathsf{E}_1,\mathsf{E}_2); \label{iqB3} \\
I_{q_{1}}(2;4) = H(\mathsf{R}_2) - H(\mathsf{E}_2,\mathsf{E}_3). \label{iqB4}
\end{align}
The rest of the proof then simply follows as before as long as the entropies in the second terms in each of the above equations comes from an environment state $\rho_{E_1,E_2,E_3}$. To ensure this we also require that $\psi_i$ must be purification of $\rho_i$, for $i\in\{1,2\}$.

Now, let us consider the case where interventions $\mathrm{A}_{j}$ are arbitrary (CPTP maps). Let $\alpha_{j}$ be the Hilbert space associated with the environmental system related to the isometric extension of $\mathrm{A}_{j}$. Then the second terms in the last set of equations becomes 
\begin{align}
H(\mathsf{E}_{1},\alpha_2,\mathsf{E}_{2},\alpha_3,\mathsf{E}_{3},\alpha_{4}) \to H(\tilde{\mathsf{E}}_{1},\tilde{\mathsf{E}}_{2},\tilde{\mathsf{E}}_{3}); \label{d1}\\
H(\mathsf{E}_2,\alpha_3) \to H(\tilde{\mathsf{E}}_2); \label{d2}\\
H(\mathsf{E}_{1},\alpha_2, \mathsf{E}_{2}, \alpha_3) \to H(\tilde{\mathsf{E}}_{1}, \tilde{\mathsf{E}}_{2}); \label{d3}\\
H(\mathsf{E}_{2}, \alpha_3, \mathsf{E}_{3},\alpha_{4}) \to H(\tilde{\mathsf{E}}_{2}, \tilde{\mathsf{E}}_{3}). \label{d4}
\end{align}
Here, we have redefine $\tilde{\mathsf{E}}_{1} :=\mathsf{E}_{1} \otimes \alpha_2$, $\tilde{\mathsf{E}}_{2} :=\mathsf{E}_{2} \otimes \alpha_3$, and $\tilde{\mathsf{E}}_{3} :=\mathsf{E}_{3} \otimes \alpha_4$. Once we absorbed $\alpha_1$ into the initial state, the rest of the proof then simply follows as before, using the strong subadditivity of quantum entropy.
\end{proof}

\begin{figure}[t!]
\begin{center}
\includegraphics[width=0.45\textwidth] {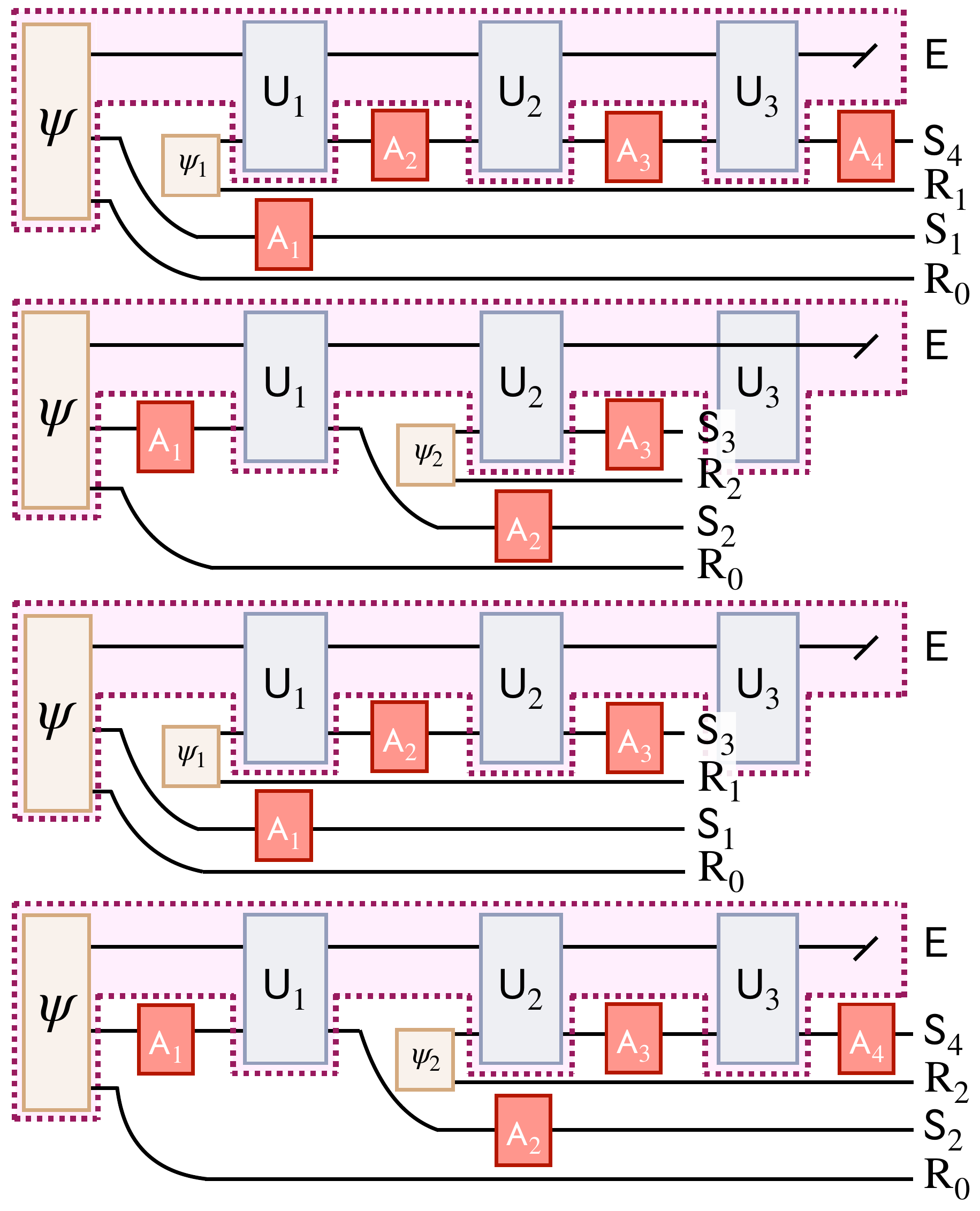}
    \caption{\textbf{Quantum Markov monogamy information with memory.} The variations in coherent information given in Eqs.~(\ref{eq:qchi-v1},\ref{eq:qchi-v2},\ref{eq:qchi-v3}) account for the memory due to the past. This figure displays how the spaces that one needs to account for.} \label{fig:processM4}
\end{center}
\end{figure}

One might now wonder, why we had to redefine coherent information in Eq.~\eqref{eq:qchi-v1} to accommodate interventions. The reason is not new. If one wants to operationally witness QMMI given in Eq.~\eqref{monogamy}, then the state of the system must be swapped with an identical copy whose purification remains in our possession. We can of course do the same here to redefine coherent information as
\begin{gather}\label{eq:qchi-v2}
I_{q_{2}}(j;k) := H(\mathsf{S}_k) - H(\mathsf{S}_j,\mathsf{R}_j,\mathsf{S}_k).
\end{gather}
Once again, $\psi_i$ must be purification of $\rho_i$, for $i\in\{1,2\}$. The big difference is that we now also account for the entropy of space $\mathsf{S}_j$, which, for non-Markovian processes, via initial correlations plays a non-trivial role. In this sense, the last equation accounts for not just tripartite correlations~\cite{Modiscirep, sai}, as before we only dealt with bipartite correlations. Yet, another possibility is the following 
\begin{gather}\label{eq:qchi-v3}
I_{q_{3}}(j;k) := H(\mathsf{S}_j,\mathsf{S}_k) - H(\mathsf{S}_j,\mathsf{R}_j,\mathsf{S}_k).
\end{gather}
Here, too we must have that $\psi_i$ must be purification of $\rho_i$, for $i\in\{1,2\}$. And here, again, we account for tripartite correlations thus this a more powerful version of MQMMI. The proof for the MQMMI for the above two equations follows the same path as the last two proofs and we omit the details. Note that we could have kept $\mathsf{R}_0$ in the above definition, and there are many other alternatives.

A couple of remarks are in order at this stage. Firstly, all of these MQMMI yield the same value for a Markov processes, including the QMMI inequality from the previous section. All three version of the MQMMI above require that a state fed into the process at an intermediate stage must be a purification of the previous output state. This is required so that the strong-subaditivity inequalities can be applied. While this is the same requirement as the QMMI inequality in the previous section, the last three version should be able to account for non-Markovianity even in divisible processes~\cite{milz-kim}. This is because, they are designed to account for multitime correlations.

We can just as well construct multitime QDPIs with the three definitions of coherent information above. In fact, one can also use the above definitions in the classical case. The key point is that these constructions account for multitime correlations by allowing for multitime entropies. In contrast to the classical case, multitime entropies will be stronger indicators of non-Markovianity when quantum entanglement in time~\cite{arXiv:1811.03722, 10.21468/SciPostPhys.10.6.141, White2021b} is present in the process, which may serve as an important diagnostic tool.

\subsection{Violation of the Multitime Quantum Markov monogamy inequalities}

Now we consider the interventional approach present in the MQMMIs in order to witness the non-Markovian beheviour of the process represented in Fig.~\ref{NMdiagram}. 

Allowing for interventions, the non-Markovian process represented in Fig.~\ref{NMdiagram} is then described by the process tensor represented in Fig.~\ref{fig:processtensor}. Furthermore, we set here a process tensor according to the bottom panel of Fig.~\ref{fig:processtensor} by defining the initial state $\psi$ as in Eq.~(\ref{initialstate}) and joint system-environment operations $\mathrm{U}_i(\bullet)=U_{\lambda}(\bullet)U_{\lambda}^{\dagger}$, for $i=1,2,3$, with $U_{\lambda}$ defined in Eq.~(\ref{UnitaryOperation}).

The MQMMIs defined with respect to the quantities in Equations~(\ref{eq:qchi-v1},\ref{eq:qchi-v2},\ref{eq:qchi-v3}) lead to the definition of the interventional witnesses of non-Markovianity:
\begin{equation} 
\mathrm{M4}_{q_i} \coloneqq I_{q_{i}}(1;4) + I_{q_{i}}(2;3) - I_{q_{i}}(1;3) -I_{q_{i}}(2;4),     \label{intwitness} 
\end{equation}
with $i=1,2,3$. The quantities defined above are positive semi-definite for any Markov process tensor, independently of the  interventional scheme adopted. Thus, finding a negative value for any $\mathrm{M4}_{q_i}$ (i=1,2,3) implies the process is non-Markovian. In Fig.~\ref{fig:processintM4} we present the plot for the witnesses in Eq.~(\ref{intwitness}) for the process tensor mentioned above.

\begin{figure}[t!]
\begin{center}
\includegraphics[width=0.5\textwidth] {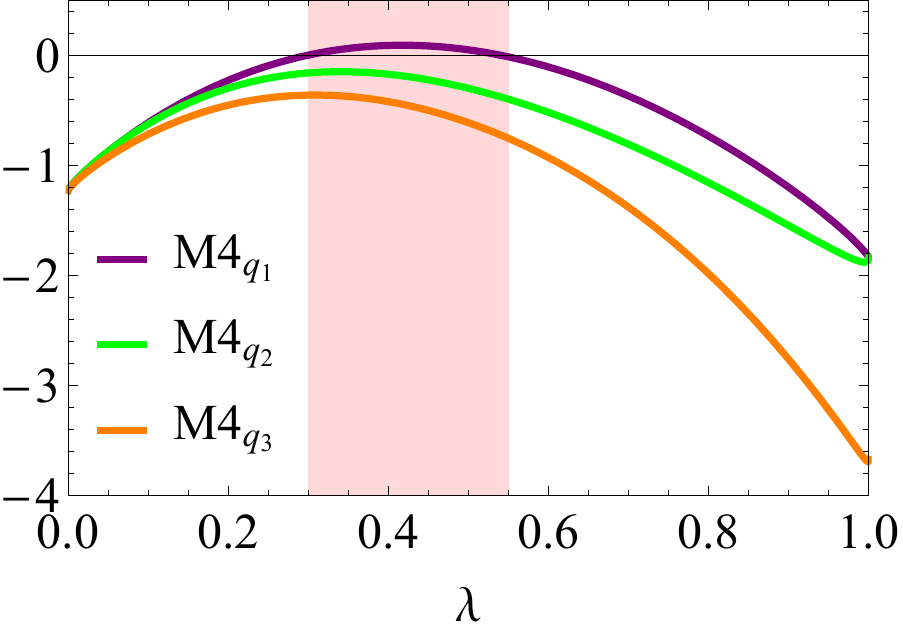}
    \caption{\textbf{Violation of the MQMMIs.} Multitime quantum Markov monogamy with respect to $I_{q_{1}}$ is not violated only for the region $0.30 \leq \lambda \leq 0.55$. The MQMMIs with respect to $I_{q_{2}}$ and $I_{q_{3}}$ are violated for any value of $\lambda$, thus perfectly witnessing the non-Markovianity of the process considered.} \label{fig:processintM4}
\end{center}
\end{figure}

\section{Conclusion} \label{conclusion}

Firstly, using the approach introduced in~\cite{schumacher1996quantum} we have extended the conjecture on classical Markov monogamy inequalities to the quantum case. We also proved that the Markov monogamy inequalities may be violated by a non-Markovian quantum process satisfying all the quantum data processing inequalities. This is done by considering concrete examples.

Secondly, we have also considered how to apply the novel quantum Markov monogamy inequalities to the process tensor formalism, which accounts for multitime quantum correlations. This provides an interesting interventional approach to quantum data processing, and thus adding extra relevance to the results.

The resources involved within witnessing non-Markovianity with QMMIs and MQMMIs are not equivalent. Generally, characterizing non-Markov phenomena with QMMIs involves assessing reference-environment systems. Distinctly, the MQMMIs do not requires this strong requirement, and only depend upon the system's properties. Nevertheless, it requires feeding intermediate steps of the process with the purification of the previous state of the system. Thus, we claim it is not fair to compare the results relying this two approaches on an equal footing.

Quantum data processing theorems have been widely studied in the theory of quantum information. In particular, the QDPI in Eq.~\eqref{qdp} has been shown to have a relevant interpretation: the local evolution of a quantum system can be cast as a completely positive and trace preserving operation if and only if the QDPI is satisfied~\cite{buscemi2014complete}. A stronger operation interpretation for these new inequalities is still missing so far, and is left for future studies.

Finally, it is worthy mentioning that the quantum information inequalities we discovered have the potential to be related with novel recovery operations. This is supported by the existence of an operational interpretation of the quantum data processing inequality. In the original formulation in Ref.~\cite{schumacher1996quantum}, the authors used the quantum data processing inequality to show that, given an initial state $\rho$ and a quantum channel $\Lambda$, there is a recovery operation $\mathrm{R}$ such that $\mathrm{R} ( \Lambda(\rho) ) = \rho$ if and only if $I_c(\rho;\mathrm{id})=I_c(\rho;\Lambda)$. Therefore, this suggests the novel QMMIs could be also related to limits on the quantum processing of information. One possible direction would be providing similar operational meaning to the quantum Markov monogamy inequalities in terms of memory strength and quantum recovery procedure~\cite{arXiv:1907.05807}. We leave the study of operational properties of the Markov monogamy inequalities to future work.

\begin{acknowledgments}
MC and LCC acknowledge the financial support from funding agencies CNPq, FAPEG and the Brazilian National Institute of Science and Technology of Quantum Information (INCT- IQ). This study was financed in part by the Coordenação de Aperfeiçoamento de Pessoal de Nível Superior – Brasil (CAPES) – Finance Code 001. MC also acknowledges the warm hospitality of the School of Mathematics and Physics at the University of Queensland, where initial part of this study was developed. Special thanks are expressed to Kaumudibikash Goswami, with whom MC had inspiring discussions. RC acknowledge the support of the John Templeton Foundation (grant Q-CAUSAL No 61084), the Serrapilheira Institute (grant number Serra – 1708-15763), CNPq via the INCT-IQ and Grants No. 406574/2018-9 and 307295/2020-6, the Brazilian agencies MCTIC and MEC.
KM is supported through Australian Research Council Future Fellowship FT160100073, Discovery Projects DP210100597 \& DP220101793, and the International Quantum U Tech Accelerator award by the US Air Force Research Laboratory.
\end{acknowledgments}

\appendix

\section*{Appendix}

\subsection{Proof of the data processing theorem and its equivalence with the monotonicity condition of quantum mutual information} \label{CohInftoMutInf}

We start by considering the proof of the quantum data processing theorem in \cite{schumacher1996quantum}.

\begin{theorem}[Quantum data processing inequality~\cite{schumacher1996quantum}] \label{qdpthm}
Let $\rho$ be a quantum state of $\mathsf{S}_{1}$, and $\Lambda_{1} \colon \mathsf{L}(\mathsf{S}_{1})\rightarrow \mathsf{L}(\mathsf{S}_{2}),\Lambda_{2} \colon \mathsf{L}(\mathsf{S}_{2})\rightarrow \mathsf{L}(\mathsf{S}_{3})$ be quantum channels. It holds that
\begin{equation} \label{qdpineq}
I_{c}(\rho;\Lambda_{1}) \geq I_{c}(\rho;\Lambda_{2}\circ \Lambda_{1}). 
\end{equation}
\end{theorem}
\begin{proof}
Let the state $\rho$ be purified to $\psi$ in $\mathsf{L}(\mathsf{R} \otimes \mathsf{S}_{1})$. Let the quantum channels be dilated according to 
\begin{equation} \label{dilation1}
\Lambda_{1}(\sigma_{1})= \Tr_{\mathsf{E}_{1}}\left[ U_{1}(\sigma_{1} \otimes \varphi_{1})U_{1}^{\dagger} \right]    
\end{equation}
and 
\begin{equation} \label{dilation2}
\Lambda_{2}(\sigma_{2})= \Tr_{\mathsf{E}_{2}}\left[ U_{2}(\sigma_{2} \otimes \varphi_{2})U_{2}^{\dagger} \right],    
\end{equation}
for any operators $\sigma_{1}$ and $\sigma_{2}$ in $\mathsf{L}(\mathsf{S}_{1})$ and $\mathsf{L}(\mathsf{S}_{2})$, respectively. The linear transformations $U_{1} \colon \mathsf{S}_{1} \otimes \mathsf{F}_1 \rightarrow \mathsf{S}_{2} \otimes \mathsf{E}_1$ and $U_{2} \colon \mathsf{S}_{2} \otimes \mathsf{F}_2 \rightarrow \mathsf{S}_{3} \otimes \mathsf{E}_2$ are unitary operators, and the pure quantum states $\varphi_{1}$ and $\varphi_{2}$ are in $\mathsf{L}(\mathsf{F}_{1})$ and $\mathsf{L}(\mathsf{F}_{2})$, respectively. Consider the following mathematical assertions with respect to the process represented in Fig.~\ref{DPIdiagram}.

\begin{equation} \label{A1}
\mathsf{R} \otimes \mathsf{E}_{1} \otimes \mathsf{E}_{2} \otimes \mathsf{S}_{3} \, \text{is pure} \Rightarrow H(\mathsf{R} , \mathsf{E}_{1} , \mathsf{E}_{2} )=H(\mathsf{S}_{3});
\end{equation}
\begin{equation} \label{A2}
\mathsf{R} \otimes \mathsf{E}_{1} \otimes \mathsf{S}_{2} \, \text{is pure} \Rightarrow \; H(\mathsf{R} , \mathsf{E}_{1})=H(\mathsf{S}_{2});
\end{equation}
\begin{equation} \label{A3}
H(\mathsf{R},\mathsf{S}_{2})=H(\mathsf{E}_{1});
\end{equation}
\begin{equation} \label{A4}
H(\mathsf{R},\mathsf{S}_{3})=H(\mathsf{E}_{1},\mathsf{E}_{2}).
\end{equation}

\begin{figure}   
\begin{adjustbox}{width=0.4\textwidth}
\begin{tikzpicture}

\draw[ultra thick,red, fill=red!10] (0,0) rectangle (4,2);
\draw[ultra thick, CadetBlue, fill=CadetBlue!10] (3,4) rectangle (7,6);
\draw[ultra thick, CadetBlue, fill=CadetBlue!10] (6,8) rectangle (10,10);
\draw[ultra thick, OliveGreen, fill=OliveGreen!10] (5.5,0) rectangle (7.5,2);
\draw[ultra thick, brown, fill=brown!10] (8.5,0) rectangle (10.5,2);
\draw[ultra thick, red] (3.5,2) -- (3.5,4);\draw[ultra thick, OliveGreen] (3.5,6) -- (3.5,12);
\draw[ultra thick, OliveGreen] (6.5,2) -- (6.5,4); \draw[ultra thick, red] (6.5,6) -- (6.5,8);\draw[ultra thick, brown] (6.5,10) -- (6.5,12);
\draw[ultra thick, brown] (9.5,8) -- (9.5,2); \draw[ultra thick, red] (9.5,10) -- (9.5,12);
\draw[ultra thick, black] (0.5,2) -- (0.5,12);

\node[scale=2.5] at (2,1) {$\psi$};
\node[scale=2.5] at (6.5,1) {$\varphi_{1}$};
\node[scale=2.5] at (9.5,1) {$\varphi_{2}$};
\node[scale=2.5] at (5,5) {$U_{1}$};
\node[scale=2.5] at (8,9) {$U_{2}$};

\draw[dashed] (0,5) -- (3,5); \draw[dashed] (7,5) -- (10.5,5);
\draw[dashed] (0,9) -- (6,9); \draw[dashed] (10,9) -- (10.5,9);

\node[scale=2.5] at (1.2,3) {$\mathsf{R}$}; \node[scale=2.5] at (1.2,7) {$\mathsf{R}$};
\node[scale=2.5] at (1.2,10.9) {$\mathsf{R}$};
\node[scale=2.5] at (4.2,3) {$\mathsf{S}_{1}$};
\node[scale=2.5] at (7.2,7) {$\mathsf{S}_{2}$};
\node[scale=2.5] at (10.2,11) {$\mathsf{S}_{3}$};
\node[scale=2.5] at (7.2,3) {$\mathsf{F}_{1}$};
\node[scale=2.5] at (4.2,6.9) {$\mathsf{E}_{1}$};
\node[scale=2.5] at (4.2,10.9) {$\mathsf{E}_{1}$};
\node[scale=2.5] at (10.2,3.4) {$\mathsf{F}_{2}$};
\node[scale=2.5] at (10.2,6.9) {$\mathsf{F}_{2}$};
\node[scale=2.5] at (7.2,10.9) {$\mathsf{E}_{2}$};

\end{tikzpicture}
\end{adjustbox}
    \caption{\textbf{Diagram representing the purified process $\rho_{1} \myeqA \rho_{2} \myeqB \rho_{3}$.} The diagram displays the pure final state obtained by acting successively the isometric representations $U_{i}(\bullet \otimes \ket{\varphi_{i}})$ of $\Lambda_{i}$ (with $i=1,2$) on the purification $\psi$ of $\rho$.}
      \label{DPIdiagram}
\end{figure}
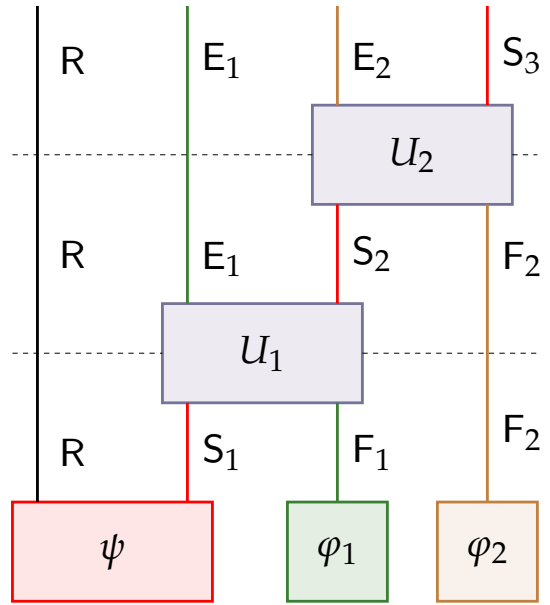

Then, we have that the strong subadditivity inequality 
\begin{equation} 
	H(\mathsf{R} , \mathsf{E}_{1} , \mathsf{E}_{2} )+H(\mathsf{E}_{1}) \leq H(\mathsf{R} , \mathsf{E}_{1})+H(\mathsf{E}_{1} , \mathsf{E}_{2})
\end{equation}
and Equations (\ref{A1},\ref{A2},\ref{A3},\ref{A4}) imply
the desired data processing inequality.
\end{proof}

Now we are ready to prove the equivalence of Theorem \ref{qdpthm} with the monotonicity of quantum mutual information under local operations. Nevertheless, we state first a simple result considered in this derivation.
\begin{lemma}[Proposition 2.29 in~\cite{watrous2018theory}] \label{Lemma1}
Let $\psi$ be a pure state of $\mathsf{R} \otimes \mathsf{S}_{1}$, and let $\rho$ be any state  of $\mathsf{R} \otimes \mathsf{S}_{2}$ for which $\Tr_{\mathsf{S}_{1}}[\psi]=\Tr_{\mathsf{S}_{2}}[\rho]$. Then there is a quantum channel $\Lambda \colon \mathsf{L}(\mathsf{S}_{1}) \rightarrow \mathsf{L}(\mathsf{S}_{2})$ such that $\rho=(\mathrm{id}_{\mathsf{R}}\otimes\Lambda)(\psi)$.
\end{lemma}

\begin{proposition} \label{Prop1}
The following sentences are equivalent:
\begin{enumerate}[(A)]%[I] for capital roman numbers.
\item The data processing inequality (\ref{qdpineq}) holds for any quantum state $\rho$ of $\mathsf{S}_{1}$, and for any quantum channels $\Lambda_{1} \colon \mathsf{L}(\mathsf{S}_{1})\rightarrow \mathsf{L}(\mathsf{S}_{2})$ and $\Lambda_{2} \colon \mathsf{L}(\mathsf{S}_{2})\rightarrow \mathsf{L}(\mathsf{S}_{3})$;
\item For any quantum state $\sigma$ of a bipartite system $\mathsf{A} \otimes \mathsf{B}$, and for any quantum operation $\Lambda \colon \mathsf{L}(\mathsf{B})\rightarrow \mathsf{L}(\mathsf{C})$, it holds that quantum mutual information is monotonically decreasing under the action of the local operation $\Lambda$. That is,
\begin{equation}
I(\mathsf{A}:\mathsf{B})_{\sigma} \geq I(\mathsf{A}:\mathsf{C})_{(\mathrm{id}\otimes\Lambda)(\sigma)}.
\end{equation}
\end{enumerate}
\end{proposition}
\begin{proof}
The assertion $(B) \Rightarrow (A)$ is clearly true. Suppose $(A)$ is true. Then let $\rho$ be the state of an arbitrary quantum system $\mathsf{S}_{1}$, and $\psi$ be a purification with respect to a bipartite system $\mathsf{R} \otimes \mathsf{S}_{1}$. Consider also arbitrary quantum channels $\Lambda_{1} \colon \mathsf{L}(\mathsf{S}_{1}) \rightarrow \mathsf{L}(\mathsf{S}_{2})$ and $\Lambda_{2} \colon \mathsf{L}(\mathsf{S}_{2})\rightarrow \mathsf{L}(\mathsf{S}_{3})$. Since $(A)$ is true by hypothesis, and $\sigma \coloneqq (\mathrm{id}_{\mathsf{R}} \otimes \Lambda_{1}) (\psi)$ is a state of the bipartite system $\mathsf{R} \otimes \mathsf{S}_{2}$, we have
\begin{equation} \label{MImon1}
I(\mathsf{R} : \mathsf{S}_{2})_{\sigma} \geq I(\mathsf{R} : \mathsf{S}_{2})_{(\mathrm{id}_{\mathsf{R}} \otimes \Lambda_{2}) (\sigma)}.
\end{equation}
Thus, subtracting $H(\mathsf{R})$ from both sides of Eq.~(\ref{MImon1}) we have the desired inequality
\begin{equation}\label{MImon2}
I_{c}(\rho;\Lambda_{1}) \geq I_{c}(\rho;\Lambda_{2} \circ \Lambda_{1}),
\end{equation}
for arbitrary $\rho$, $\Lambda_{1}$ and $\Lambda_{2}$.

Now, let us prove $(A) \Rightarrow (B)$ is true. So suppose $(A)$ is true. Following the derivation of $(B) \Rightarrow (A)$ we see that in order to prove its converse statement all we need to do is to prove that all bipartite quantum states can be written as $\sigma \coloneqq (\mathrm{id}_{\mathsf{R}} \otimes \Lambda_{2}) (\psi)$ for some pure state $\psi$ of a bipartite system, and quantum channel $\Lambda_{1}$. Then, we can add $H(\mathsf{R})$ to both sides of Eq.~(\ref{MImon2}) and we are done. So let $\rho$ be an arbitrary quantum state of any bipartite quantum system $\mathsf{R} \otimes \mathsf{S}_{2}$. Take its marginal with respect to the system $\mathsf{R}$, that is, $\tau \coloneqq \Tr_{\mathsf{R}}[\rho]$. Now let $\psi$ be a purification of $\tau$ with respect to a purification system $\mathsf{S}_{1}$ such that $\psi \in \mathsf{R} \otimes \mathsf{S}_{1}$. Thus we have proved the existence of a pure quantum state $\psi$ for which $\Tr_{\mathsf{S}_{1}}[\psi]=\Tr_{\mathsf{S}_{2}}[\rho]$. Now, using Lemma~\ref{Lemma1} we make sure the existence of a quantum channel $\Lambda_{1} \colon \mathsf{L}(\mathsf{S}_{1})\rightarrow \mathsf{L}(\mathsf{S}_{2})$ fulfilling the desired property. 
\end{proof}

\subsection{Quantum data processing inequalities for four-time-step Markovian processes} \label{DPIextraAppendix}
We argue here on the possible validity of the data processing inequalities not considered in Subsection~\ref{QMMViolation}. We also show that the DPIs not appearing in Fig.~\ref{monogamy_violation} do not witness the non-Marovianity of our example, and thus are irrelevant in this case.

One would expect for each DPI to have a valid quantum version in terms of coherent information. For instance, it is expected that the inequality
\begin{equation} \label{possibleDPI}
I_{c}(\Lambda_{1}(\rho);\Lambda_{2}) \geq I_{c}(\rho;\Lambda_{2} \circ \Lambda_{1})    
\end{equation}
constrains three-time-step quantum processes. This would be the quantum version of the CDPI given by $I(X_{2}:X_{3}) \geq I(X_{1}:X_{3})$, holding for any three-time-step classical process. 
Equation~(\ref{possibleDPI}) is clearly equivalent to the condition
\begin{equation} \label{NonTrivialDPi}
   H(\mathsf{E}_{1}|\mathsf{E}_{2}) \geq 0,
\end{equation}
where conditional entropy is computed on any quantum state of the form presented in Fig.~\ref{DPIdiagram}.

It is well known that the conditional quantum entropy may be negative, in contrast to its classical counterpart. This is indeed the case when the quantum state considered is maximally entangled, for instance. Although, it is not clear that this behavior may appear from the quantum state arising from a Markov process. The plot of Fig.~\ref{DPtest} shows that taking $U_{1},U_{2}$ to be the operation of Eq.~(\ref{UnitaryOperation}), and quantum states $\psi=\Psi^{+}$, $\varphi_{1}=\varphi_{2}=\ket{0}$, the condition in Eq.~(\ref{possibleDPI}) is satisfied for any value of $0 \leq \lambda \leq 1$. That is, the quantity 
\begin{equation} \label{witness6}
    \mathrm{DP}_{5} \coloneqq I_{c}(\Lambda_{1}(\rho);\Lambda_{2}) - I_{c}(\rho;\Lambda_{2} \circ \Lambda_{1})
\end{equation}
is positive for any value of $\lambda$.

We note there is no strong subaditivity solely implying Eq.~(\ref{NonTrivialDPi}), and we leave its proof as a future study.

\begin{figure}[t!]
\begin{center}

\includegraphics[width=0.5\textwidth]{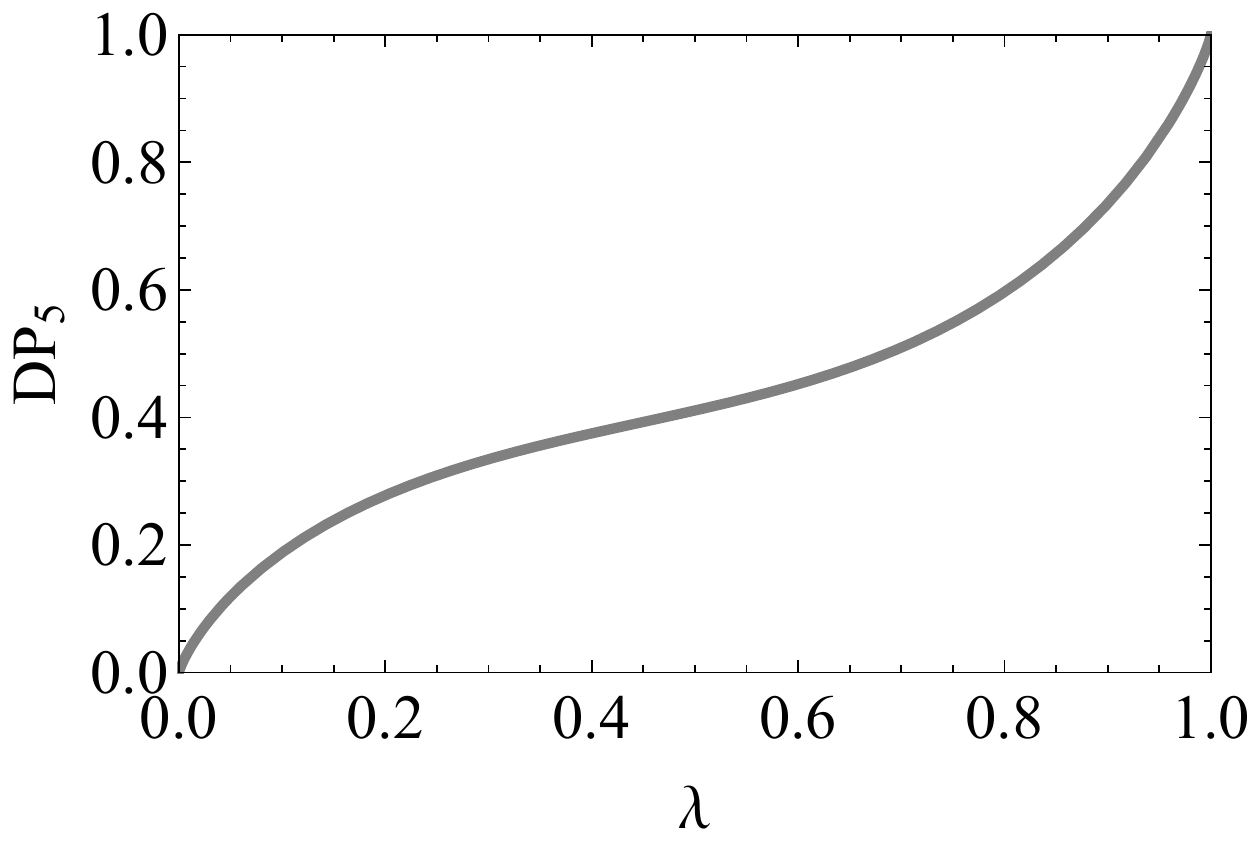}
    \caption{\textbf{Non-violation of $\mathrm{DP}_5$.} The quantity defined in Eq.~(\ref{witness6}) is non-negative for the Markov process in Fig.~\ref{DPIdiagram} with $\psi$ being the maximally entangled state, and $U_1,U_2$ the operator defined in Eq.~(\ref{UnitaryOperation}). } \label{DPtest}
\end{center}
\end{figure}

In Appendix~\ref{CohInftoMutInf}, we have proved the quantities defined in Equations~(\ref{witness1}-\ref{witness5}) are non-negative for all four-time-step quantum Markov processes. Now, supposing every classical DPI has a valid quantum version, we would have the following extra constraints:
\begin{align} 
\mathrm{DP}_{6} \coloneqq& I_{c}(\Lambda_{1}(\rho);\Lambda_{2}) - I_{c}(\rho;\Lambda_{3} \circ \Lambda_{2} \circ \Lambda_{1});\label{witness7} \\
\mathrm{DP}_{7} \coloneqq& I_{c}(\Lambda_{1}(\rho);\Lambda_{3} \circ \Lambda_{2}) - I_{c}(\rho;\Lambda_{3} \circ \Lambda_{2} \circ \Lambda_{1}) ;\label{witness8} \\
\mathrm{DP}_{8} \coloneqq& I_{c}(\Lambda_{2} \circ \Lambda_{1}(\rho);\Lambda_{3}) - I_{c}(\rho;\Lambda_{3} \circ \Lambda_{2} \circ \Lambda_{1}) ;\label{witness9}\\
\mathrm{DP}_{9} \coloneqq& I_{c}(\Lambda_{2} \circ \Lambda_{1}(\rho);\Lambda_{3}) - I_{c}(\Lambda_{1}(\rho);\Lambda_{3} \circ \Lambda_{2}).\label{witness10}
\end{align}

Moreover, for the example considered in Subsection~\ref{QMMViolation}, the DPIs in Equations~(\ref{witness6},\ref{witness7},\ref{witness8}) do not witness the non-Markovian behaviour of the process. See Fig.~\ref{DPIextra}. The remaining inequalities in Equations~(\ref{witness9},\ref{witness10}) involve terms related to a bipartite environmental system, and thus cannot be applied to our example.

\begin{figure}[t!]
\begin{center}

\includegraphics[width=0.5\textwidth]{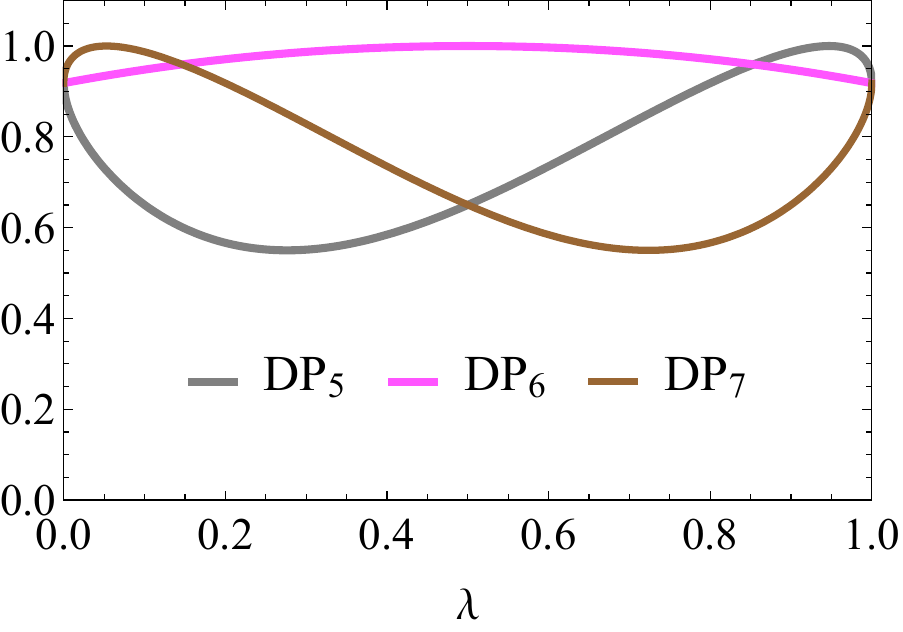}
    \caption{\textbf{Non-violation of $\mathrm{DP_{5}}$, $\mathrm{DP_{6}}$ and $\mathrm{DP_{7}}$ for the non-Markov process considered in Subsec.~\ref{QMMViolation}.} The quantities defined in Eq.s~(\ref{witness6},\ref{witness7},\ref{witness8}) are positive for the processes described by Eq.s~(\ref{gamma1},\ref{gamma2},\ref{gamma3},\ref{gamma4}) and represented in Fig.~\ref{NMdiagram}.} \label{DPIextra}
\end{center}
\end{figure}

\subsection{Quantum Markov monogamy theorem as the monotonicity of conditional quantum mutual information} \label{M4toCQMI}

Now the Markov monogamy inequalities are subjected to close scrutiny. This section deals with the monogamy inequality of four-time-step Markovian processes.

\begin{proposition} \label{propositionFour}
The following sentences are equivalent:
\begin{enumerate}[(A)]%[I] for capital roman numbers.
\item The Monogamy inequality (\ref{monogamy}) holds for any quantum state $\rho$ of $\mathsf{S}_{1}$, and for any quantum channels $\Lambda_{1} \colon \mathsf{L}(\mathsf{S}_{1})\rightarrow \mathsf{L}(\mathsf{S}_{2})$, $\Lambda_{2} \colon \mathsf{L}(\mathsf{S}_{2})\rightarrow \mathsf{L}(\mathsf{S}_{3})$ and $\Lambda_{3} \colon \mathsf{L}(\mathsf{S}_{3})\rightarrow \mathsf{L}(\mathsf{S}_{4})$;
\item For any quantum state $\sigma$ of a tripartite system $\mathsf{A} \otimes \mathsf{B} \otimes \mathsf{C}$, and for any quantum operation $\Lambda \colon \mathsf{L}(\mathsf{B})\rightarrow \mathsf{L}(\mathsf{D})$, it holds that conditional quantum mutual information is monotonically decreasing under the action of a local operation $\Lambda$. That is,
\begin{equation}
I(\mathsf{A}:\mathsf{B}|\mathsf{C})_{\sigma} \geq I(\mathsf{A}:\mathsf{D}|\mathsf{C})_{(\mathrm{id}\otimes\Lambda\otimes\mathrm{id})(\sigma)}.
\end{equation}
\end{enumerate}
\end{proposition}
\begin{proof}
The proof is similar to the one of Proposition~\ref{Prop1}. The sentence $(B) \Rightarrow (A)$ is trivially true. To prove $(A) \Rightarrow (B)$ we need to show that any tripartite quantum state $\rho$ can be written as $\rho=(\mathrm{id}_{\mathsf{A}}\otimes\Lambda\otimes\mathrm{id}_{\mathsf{C}})(\sigma)$, with $\sigma=(\mathrm{id}_{\mathsf{A}} \otimes \Phi) (\psi \otimes \varphi)$, where $\psi$ is a state of a bipartite system $\mathsf{A} \otimes \mathsf{D}$, $\varphi$ is a pure bipartite state of a system $\mathsf{E}$, $\Phi \colon \mathsf{L}(\mathsf{D} \otimes \mathsf{E}) \rightarrow \mathsf{L}(\mathsf{F} \otimes \mathsf{C})$ is a unitary quantum channel, and $\Lambda \colon \mathsf{L}(\mathsf{F})\rightarrow \mathsf{L}(\mathsf{B})$ is a quantum channel. It can be done by using Lemma~\ref{Lemma1} twice as follows.

Let $\tau$ be the tripartite quantum state obtained by swapping the systems $\mathsf{B}$ and $\mathsf{C}$ of $\rho$. Then define the marginal state with respect to the bipartite system $\mathsf{A \otimes C}$, $\omega\coloneqq\Tr_{\mathsf{B}}[\tau]$. By the same reasoning as in the proof of Proposition~\ref{Prop1}, we know there are a pure bipartite state $\psi \in \mathsf{L}(\mathsf{A} \otimes \mathsf{D})$ and a quantum channel $\Omega \coloneqq \mathsf{L}(\mathsf{D})\rightarrow \mathsf{L}(\mathsf{C})$ for which $\omega = (\mathrm{id}_{\mathsf{A}}\otimes\Omega) (\psi)$. Let $U \colon \mathsf{D} \otimes \mathsf{E} \rightarrow \mathsf{C} \otimes \mathsf{F}$ be a dilation of $\Omega$, such that for some pure state $\varphi$ of the system $\mathsf{E}$ we have $\Omega(\bullet)=\Tr_{\mathsf{F}}[U(\bullet \otimes \varphi)U^{\dagger}]$. Then define the unitary quantum channel $\widetilde{\Phi}(\bullet)=V(\bullet \otimes \varphi)V^{\dagger}$, where $V$ is the unitary operator obtained by the action of $U$ followed by the swapping operation. Define the pure tripartite quantum state $\eta= (\mathrm{id}_{\mathsf{A}}\otimes\widetilde{\Phi}) (\psi \otimes \varphi)$. The state $\eta$ is a purification of $\omega$. Then we have $\Tr_{\mathsf{F}}[\eta]=\Tr_{\mathsf{B}}[\tau]$. Thus, by Lemma~\ref{Lemma1}, there is a quantum channel $\Lambda \colon \mathsf{L}(\mathsf{F})\rightarrow \mathsf{L}(\mathsf{B})$ for which $\tau=(\mathrm{id}_{\mathsf{A}} \otimes \mathrm{id}_{\mathsf{C}}\otimes\Lambda) (\eta)$. Moreover, take the swapping of systems $\mathsf{C}$ and $\mathsf{B}$ of $\tau$ to recover 
\begin{equation}
\rho=(\mathrm{id}_{\mathsf{A}} \otimes \Lambda \otimes \mathrm{id}_{\mathsf{C}})(\mathrm{id}_{\mathsf{A}} \otimes \Phi) (\psi \otimes \varphi),
\end{equation}
with $\Phi(\bullet)=U (\bullet) U^{\dagger}$.
\end{proof}
\subsection{Quantum Markov monogamy inequalities for $n=3$}\label{M6}

The following results deal with the Markov monogamy conditions for six-time-step Markov processes.

\begin{theorem} \label{Monogamy6} For any Markov process 
\begin{equation*}
    \rho_1 \myeqA \rho_2 \myeqB \cdots \myeqfive \rho_6, 
\end{equation*}
it holds the following inequalities:
\begin{widetext}
\begin{eqnarray}
		  I_{c}(\rho_{1}:\rho_{6})+I_{c}(\rho_{2}:\rho_{5})+I_{c}(\rho_{3}:\rho_{4}) 
  		  \geq
  		  I_{c}(\rho_{1}:\rho_{4})+I_{c}(\rho_{2}:\rho_{6})+I_{c}(\rho_{3}:\rho_{5});\label{M6a}\\
  		  I_{c}(\rho_{1}:\rho_{6})+I_{c}(\rho_{2}:\rho_{5})+I_{c}(\rho_{3}:\rho_{4})    
  		  \geq
		  I_{c}(\rho_{1}:\rho_{5})+I_{c}(\rho_{2}:\rho_{4})+I_{c}(\rho_{3}:\rho_{6}).\label{M6b}  		  
\end{eqnarray}
\end{widetext}
\end{theorem}

\begin{proof}
Let the quantum channels $\Lambda_{i} \colon \mathsf{L}(\mathsf{S}_{i}) \rightarrow \mathsf{L}(\mathsf{S}_{i+1})$ -- with $i=1,\cdots,5$ -- have isometric representation given by $V_{i}:\mathsf{S}_{i} \rightarrow \mathsf{S}_{i+1} \otimes \mathsf{E}_{i}$. The proof are given by the strong subadditivity inequalities relating the environmental systems $\mathsf{E}_i$ such that added together imply the desired Monogamy inequality.

In order to prove Eq.~(\ref{M6a}) add the strong subadditivity inequalities
\begin{eqnarray}
I(\mathsf{E}_{1}:\mathsf{E}_{5}|\mathsf{E}_{2},\mathsf{E}_{3},\mathsf{E}_{4}) &\geq& 0; \\
I(\mathsf{E}_{1},\mathsf{E}_{2}:\mathsf{E}_{4}|\mathsf{E}_{3}) &\geq& 0.
\end{eqnarray}

Now, to prove Eq.~(\ref{M6b}) consider 
\begin{eqnarray}
I(\mathsf{E}_{1}:\mathsf{E}_{5}|\mathsf{E}_{2},\mathsf{E}_{3},\mathsf{E}_{4}) &\geq& 0; \\
I(\mathsf{E}_{1}:\mathsf{E}_{4},\mathsf{E}_{5}|\mathsf{E}_{3}) &\geq& 0.
\end{eqnarray}
\end{proof}
Now we consider how the Markov monogamy inequalities for six-time-step processes can be equivalently stated in terms of conditional quantum mutual information.
\begin{proposition} \label{propositionSix}
The following sentences are equivalent:
\begin{enumerate}[(A)]%[I] for capital roman numbers.
\item Theorem \ref{Monogamy6} holds;
\item For any quantum state $\rho \in  \mathsf{L} (\mathsf{R} \otimes \mathsf{E}_{1} \otimes \mathsf{E}_{2} \otimes \mathsf{S}_{4})$, and for any quantum channels $\Lambda_{4} \colon \mathsf{L} (\mathsf{S}_{4}) \rightarrow  \mathsf{L} (\mathsf{S}_{5})$ and $\Lambda_{5} \colon \mathsf{L} (\mathsf{S}_{5}) \rightarrow  \mathsf{L} (\mathsf{S}_{6})$, it holds that
\begin{eqnarray}
I( \mathsf{R} : \mathsf{S}_{4} | \mathsf{E}_{1} , \mathsf{E}_{2}) \geq I( \mathsf{R} , \mathsf{E}_{1} : \mathsf{S}_{5}| \mathsf{E}_{2}) + I( \mathsf{R} : \mathsf{S}_{6}| \mathsf{E}_{1}); \\
I( \mathsf{R} , \mathsf{E}_{1} : \mathsf{S}_{4} |  \mathsf{E}_{2}) + I( \mathsf{R} : \mathsf{S}_{5}| \mathsf{E}_{1}) \geq  I( \mathsf{R} : \mathsf{S}_{6}| \mathsf{E}_{1} , \mathsf{E}_{2}).
\end{eqnarray}
\end{enumerate}
\end{proposition}
The proof of Proposition~\ref{propositionSix} is similar to the one presented for Proposition ~\ref{propositionFour}, so we do not include it here.

\subsection{Quantum Markov monogamy inequalities for $n=4$}\label{M8}

The following Theorems deal with the Markov monogamy conditions for eight-time-step Markov processes.

\begin{theorem} \label{Monogamy8} It holds that for any Markov process 
\begin{equation*}
    \rho_1 \myeqA \rho_2 \myeqB \cdots \myeqseven \rho_8, 
\end{equation*}
\begin{widetext}
\begin{eqnarray}
    I_{c}(\rho_{1}:\rho_{8})+I_{c}(\rho_{2}:\rho_{7})+I_{c}(\rho_{3}:\rho_{6})+I_{c}(\rho_{4}:\rho_{5}) 
    \geq
    I_{c}(\rho_{1}:\rho_{5})+I_{c}(\rho_{2}:\rho_{8})+I_{c}(\rho_{3}:\rho_{7})+I_{c}(\rho_{4}:\rho_{6}) \label{M8a}; \\
%%%%%%%%%%%%%%%%%%%%%%%%%%%%%%%%%%%%%%%%%%%%%%%%%%%%%%%%%%%%%%%%%%%%%%%%%%%%%%%%%%%%%%%%%%%%%%%%%%%%%%%%%%
    I_{c}(\rho_{1}:\rho_{8})+I_{c}(\rho_{2}:\rho_{7})+I_{c}(\rho_{3}:\rho_{6})+I_{c}(\rho_{4}:\rho_{5}) 
    \geq
    I_{c}(\rho_{1}:\rho_{7})+I_{c}(\rho_{2}:\rho_{5})+I_{c}(\rho_{3}:\rho_{8})+I_{c}(\rho_{4}:\rho_{6}) \label{M8b}; \\
%%%%%%%%%%%%%%%%%%%%%%%%%%%%%%%%%%%%%%%%%%%%%%%%%%%%%%%%%%%%%%%%%%%%%%%%%%%%%%%%%%%%%%%%%%%%%%%%%%%%%%%%%%
    I_{c}(\rho_{1}:\rho_{8})+I_{c}(\rho_{2}:\rho_{7})+I_{c}(\rho_{3}:\rho_{6})+I_{c}(\rho_{4}:\rho_{5}) 
    \geq
    I_{c}(\rho_{1}:\rho_{6})+I_{c}(\rho_{2}:\rho_{8})+I_{c}(\rho_{3}:\rho_{5})+I_{c}(\rho_{4}:\rho_{7}) \label{M8c}; \\
%%%%%%%%%%%%%%%%%%%%%%%%%%%%%%%%%%%%%%%%%%%%%%%%%%%%%%%%%%%%%%%%%%%%%%%%%%%%%%%%%%%%%%%%%%%%%%%%%%%%%%%%%%    
    I_{c}(\rho_{1}:\rho_{8})+I_{c}(\rho_{2}:\rho_{7})+I_{c}(\rho_{3}:\rho_{6})+I_{c}(\rho_{4}:\rho_{5})     	\geq
    I_{c}(\rho_{1}:\rho_{5})+I_{c}(\rho_{2}:\rho_{6})+I_{c}(\rho_{3}:\rho_{8})+I_{c}(\rho_{4}:\rho_{7}) \label{M8d}; \\
%%%%%%%%%%%%%%%%%%%%%%%%%%%%%%%%%%%%%%%%%%%%%%%%%%%%%%%%%%%%%%%%%%%%%%%%%%%%%%%%%%%%%%%%%%%%%%%%%%%%%%%%%% 
    I_{c}(\rho_{1}:\rho_{8})+I_{c}(\rho_{2}:\rho_{7})+I_{c}(\rho_{3}:\rho_{6})+I_{c}(\rho_{4}:\rho_{5}) 
    \geq
    I_{c}(\rho_{1}:\rho_{7})+I_{c}(\rho_{2}:\rho_{6})+I_{c}(\rho_{3}:\rho_{5})+I_{c}(\rho_{4}:\rho_{8}) \label{M8e}; \\
%%%%%%%%%%%%%%%%%%%%%%%%%%%%%%%%%%%%%%%%%%%%%%%%%%%%%%%%%%%%%%%%%%%%%%%%%%%%%%%%%%%%%%%%%%%%%%%%%%%%%%%%%% 
    I_{c}(\rho_{1}:\rho_{8})+I_{c}(\rho_{2}:\rho_{7})+I_{c}(\rho_{3}:\rho_{6})+I_{c}(\rho_{4}:\rho_{5}) 
    \geq
    I_{c}(\rho_{1}:\rho_{6})+I_{c}(\rho_{2}:\rho_{5})+I_{c}(\rho_{3}:\rho_{7})+I_{c}(\rho_{4}:\rho_{8}) \label{M8f}; \\
%%%%%%%%%%%%%%%%%%%%%%%%%%%%%%%%%%%%%%%%%%%%%%%%%%%%%%%%%%%%%%%%%%%%%%%%%%%%%%%%%%%%%%%%%%%%%%%%%%%%%%%%%% 
    I_{c}(\rho_{1}:\rho_{8})+I_{c}(\rho_{2}:\rho_{7})+I_{c}(\rho_{3}:\rho_{6})+I_{c}(\rho_{4}:\rho_{5}) 
    \geq
    I_{c}(\rho_{1}:\rho_{5})+I_{c}(\rho_{2}:\rho_{6})+I_{c}(\rho_{3}:\rho_{7})+I_{c}(\rho_{4}:\rho_{8}) \label{M8g}.
%%%%%%%%%%%%%%%%%%%%%%%%%%%%%%%%%%%%%%%%%%%%%%%%%%%%%%%%%%%%%%%%%%%%%%%%%%%%%%%%%%%%%%%%%%%%%%%%%%%%%%%%%% 
\end{eqnarray}
\end{widetext}
\end{theorem}
\begin{proof}
To prove Eq.~(\ref{M8a}) add
\begin{eqnarray}
I(\mathsf{E}_{1}:\mathsf{E}_{7}|\mathsf{E}_{2},\mathsf{E}_{3},\mathsf{E}_{4},\mathsf{E}_{5},\mathsf{E}_{6}) &\geq& 0; \\
I(\mathsf{E}_{1},\mathsf{E}_{2}:\mathsf{E}_{6}|\mathsf{E}_{3},\mathsf{E}_{4},\mathsf{E}_{5}) &\geq& 0; \\
I(\mathsf{E}_{1},\mathsf{E}_{2},\mathsf{E}_{3}:\mathsf{E}_{5}|\mathsf{E}_{4}) &\geq& 0.
\end{eqnarray}

To prove Eq.~(\ref{M8b}) add
\begin{eqnarray}
I(\mathsf{E}_{1}:\mathsf{E}_{7}|\mathsf{E}_{2},\mathsf{E}_{3},\mathsf{E}_{4},\mathsf{E}_{5},\mathsf{E}_{6}) &\geq& 0; \\
I(\mathsf{E}_{2}:\mathsf{E}_{6},\mathsf{E}_{7}|\mathsf{E}_{3},\mathsf{E}_{4},\mathsf{E}_{5}) &\geq& 0; \\
I(\mathsf{E}_{2},\mathsf{E}_{3}:\mathsf{E}_{5}|\mathsf{E}_{4}) &\geq& 0.
\end{eqnarray}

To prove Eq.~(\ref{M8c}) add
\begin{eqnarray}
I(\mathsf{E}_{1}:\mathsf{E}_{7}|\mathsf{E}_{2},\mathsf{E}_{3},\mathsf{E}_{4},\mathsf{E}_{5},\mathsf{E}_{6}) &\geq& 0; \\
I(\mathsf{E}_{1},\mathsf{E}_{2}:\mathsf{E}_{6}|\mathsf{E}_{3},\mathsf{E}_{4},\mathsf{E}_{5}) &\geq& 0; \\
I(\mathsf{E}_{3}:\mathsf{E}_{5},\mathsf{E}_{6}|\mathsf{E}_{4}) &\geq& 0.
\end{eqnarray}

To prove Eq.~(\ref{M8d}) add
\begin{eqnarray}
I(\mathsf{E}_{1}:\mathsf{E}_{7}|\mathsf{E}_{2},\mathsf{E}_{3},\mathsf{E}_{4},\mathsf{E}_{5},\mathsf{E}_{6}) &\geq& 0; \\
I(\mathsf{E}_{2}:\mathsf{E}_{6},\mathsf{E}_{7}|\mathsf{E}_{3},\mathsf{E}_{4},\mathsf{E}_{5}) &\geq& 0; \\
I(\mathsf{E}_{1},\mathsf{E}_{2},\mathsf{E}_{3}:\mathsf{E}_{5},\mathsf{E}_{6}|\mathsf{E}_{4}) &\geq& 0.
\end{eqnarray}

To prove Eq.~(\ref{M8e}) add
\begin{eqnarray}
I(\mathsf{E}_{1}:\mathsf{E}_{7}|\mathsf{E}_{2},\mathsf{E}_{3},\mathsf{E}_{4},\mathsf{E}_{5},\mathsf{E}_{6}) &\geq& 0; \\
I(\mathsf{E}_{2}:\mathsf{E}_{6},\mathsf{E}_{7}|\mathsf{E}_{3},\mathsf{E}_{4},\mathsf{E}_{5}) &\geq& 0; \\
I(\mathsf{E}_{3}:\mathsf{E}_{5},\mathsf{E}_{6},\mathsf{E}_{7}|\mathsf{E}_{4}) &\geq& 0.
\end{eqnarray}

To prove Eq.~(\ref{M8f}) add
\begin{eqnarray}
I(\mathsf{E}_{1}:\mathsf{E}_{7}|\mathsf{E}_{2},\mathsf{E}_{3},\mathsf{E}_{4},\mathsf{E}_{5},\mathsf{E}_{6}) &\geq& 0; \\
I(\mathsf{E}_{1},\mathsf{E}_{2}:\mathsf{E}_{6}|\mathsf{E}_{3},\mathsf{E}_{4},\mathsf{E}_{5}) &\geq& 0; \\
I(\mathsf{E}_{2},\mathsf{E}_{3}:\mathsf{E}_{5},\mathsf{E}_{6},\mathsf{E}_{7}|\mathsf{E}_{4}) &\geq& 0.
\end{eqnarray}

To prove Eq.~(\ref{M8g}) add
\begin{eqnarray}
I(\mathsf{E}_{1}:\mathsf{E}_{7}|\mathsf{E}_{2},\mathsf{E}_{3},\mathsf{E}_{4},\mathsf{E}_{5},\mathsf{E}_{6}) &\geq& 0; \\
I(\mathsf{E}_{1},\mathsf{E}_{2},\mathsf{E}_{3}:\mathsf{E}_{5},\mathsf{E}_{6}|\mathsf{E}_{4}) &\geq& 0; \\
I(\mathsf{E}_{2}:\mathsf{E}_{6},\mathsf{E}_{7}|\mathsf{E}_{3},\mathsf{E}_{4},\mathsf{E}_{5}) &\geq& 0; \\
I(\mathsf{E}_3:\mathsf{E}_7|\mathsf{E}_4,\mathsf{E}_5,\mathsf{E}_6) &\geq& 0.
\end{eqnarray}
%Eq.~(\ref{M8g}) has an alternative proof, thus, revealing that the collection of strong subadditivity inequalities summing up to a Monogamy inequality is not unique. For instance, Eq.~(\ref{M8g}) is also provided by 
\end{proof}
It follows the Markov monogamy theorems for eight-time-step processes in term of conditional quantum mutual information. Again, the proof of Proposition~\ref{propositionEight} below follows the same steps as Proposition~\ref{propositionFour} and is left absent here.
\begin{proposition} \label{propositionEight}
The following sentences are equivalent:
\begin{enumerate}[(A)]%[I] for capital roman numbers.
\item Theorem \ref{Monogamy8} holds;
\item For any quantum state $\rho \in  \mathsf{L} (\mathsf{R} \otimes \mathsf{E}_{1} \otimes \mathsf{E}_{2} \otimes \mathsf{E}_{3} \otimes \mathsf{S}_{5})$, and for any quantum channels $\Lambda_{5} \colon \mathsf{L} (\mathsf{S}_{5}) \rightarrow  \mathsf{L} (\mathsf{S}_{6})$, $\Lambda_{6} \colon \mathsf{L} (\mathsf{S}_{6}) \rightarrow  \mathsf{L} (\mathsf{S}_{7})$ and $\Lambda_{7} \colon \mathsf{L} (\mathsf{S}_{7}) \rightarrow  \mathsf{L} (\mathsf{S}_{8})$, it holds that
\begin{widetext}
\begin{eqnarray}
I(\mathsf{R}:\mathsf{S}_{5}|\mathsf{E}_{1},\mathsf{E}_{2},\mathsf{E}_{3}) \geq I(\mathsf{R},\mathsf{E}_{1},\mathsf{E}_{2}:\mathsf{S}_{6}|\mathsf{E}_{3})+I(\mathsf{R},\mathsf{E}_{1}:\mathsf{S}_{7}|\mathsf{E}_{2})+I(\mathsf{R}:\mathsf{S}_{8}|\mathsf{E}_{1});\\
%A%%%%%%%%%%%%%%%%%%%%%%%%%%%%%%%%%%%%%%%%%%%%%%%%%%%%%%%%%%%%%%%%%%%%%%%%%%%%%%%%%%%%%%%%%%%%%%%%%%%%%%%%
I(\mathsf{R},\mathsf{E}_{1}:\mathsf{S}_{5}|\mathsf{E}_{2},\mathsf{E}_{3}) + I(\mathsf{R}:\mathsf{S}_{7}|\mathsf{E}_{1}) \geq I(\mathsf{R},\mathsf{E}_{1},\mathsf{E}_{2}:\mathsf{S}_{8}|\mathsf{E}_{3})+I(\mathsf{R}:\mathsf{S}_{8}|\mathsf{E}_{1},\mathsf{E}_{2});\\
%B%%%%%%%%%%%%%%%%%%%%%%%%%%%%%%%%%%%%%%%%%%%%%%%%%%%%%%%%%%%%%%%%%%%%%%%%%%%%%%%%%%%%%%%%%%%%%%%%%%%%%%%%
I(\mathsf{R},\mathsf{E}_{1},\mathsf{E}_{2}:\mathsf{S}_{5}|\mathsf{E}_{3}) + I(\mathsf{R}:\mathsf{S}_{6}|\mathsf{E}_{1},\mathsf{E}_{2}) \geq I(\mathsf{R},\mathsf{E}_{1}:\mathsf{S}_{7}|\mathsf{E}_{2},\mathsf{E}_{3})+I(\mathsf{R}:\mathsf{S}_{8}|\mathsf{E}_{1});\\
%C%%%%%%%%%%%%%%%%%%%%%%%%%%%%%%%%%%%%%%%%%%%%%%%%%%%%%%%%%%%%%%%%%%%%%%%%%%%%%%%%%%%%%%%%%%%%%%%%%%%%%%%%
I(\mathsf{R}:\mathsf{S}_{5}|\mathsf{E}_{1},\mathsf{E}_{2},\mathsf{E}_{3}) + I(\mathsf{R},\mathsf{E}_{1}:\mathsf{S}_{6}|\mathsf{E}_{2}) \geq I(\mathsf{R},\mathsf{E}_{1}:\mathsf{S}_{7}|\mathsf{E}_{2},\mathsf{E}_{3})+I(\mathsf{R}:\mathsf{S}_{8}|\mathsf{E}_{1},\mathsf{E}_{2});\\
%D%%%%%%%%%%%%%%%%%%%%%%%%%%%%%%%%%%%%%%%%%%%%%%%%%%%%%%%%%%%%%%%%%%%%%%%%%%%%%%%%%%%%%%%%%%%%%%%%%%%%%%%%
I(\mathsf{R},\mathsf{E}_{1},\mathsf{E}_{2}:\mathsf{S}_{5}|\mathsf{E}_{3}) + I(\mathsf{R},\mathsf{E}_{1}:\mathsf{S}_{6}|\mathsf{E}_{2}) + I(\mathsf{R}:\mathsf{S}_{7}|\mathsf{E}_{1}) \geq I(\mathsf{R}:\mathsf{S}_{8}|\mathsf{E}_{1},\mathsf{E}_{2},\mathsf{E}_{3});\\
%E%%%%%%%%%%%%%%%%%%%%%%%%%%%%%%%%%%%%%%%%%%%%%%%%%%%%%%%%%%%%%%%%%%%%%%%%%%%%%%%%%%%%%%%%%%%%%%%%%%%%%%%%
I(\mathsf{R},\mathsf{E}_{1}:\mathsf{S}_{5}|\mathsf{E}_{2},\mathsf{E}_{3}) + I(\mathsf{R}:\mathsf{S}_{6}|\mathsf{E}_{1},\mathsf{E}_{2}) \geq I(\mathsf{R},\mathsf{E}_{1}:\mathsf{S}_{7}|\mathsf{E}_{2}) \geq I(\mathsf{R}:\mathsf{S}_{8}|\mathsf{E}_{1},\mathsf{E}_{2},\mathsf{E}_{3});\\
%F%%%%%%%%%%%%%%%%%%%%%%%%%%%%%%%%%%%%%%%%%%%%%%%%%%%%%%%%%%%%%%%%%%%%%%%%%%%%%%%%%%%%%%%%%%%%%%%%%%%%%%%%
I(\mathsf{R}:\mathsf{S}_{5}|\mathsf{E}_{1},\mathsf{E}_{2},\mathsf{E}_{3}) + I(\mathsf{R},\mathsf{E}_{1}:\mathsf{S}_{6}|\mathsf{E}_{2}) \geq I(\mathsf{R},\mathsf{E}_{1}:\mathsf{S}_{7}|\mathsf{E}_{2}) + I(\mathsf{R}:\mathsf{S}_{8}|\mathsf{E}_{1},\mathsf{E}_{2},\mathsf{E}_{3}).
%F%%%%%%%%%%%%%%%%%%%%%%%%%%%%%%%%%%%%%%%%%%%%%%%%%%%%%%%%%%%%%%%%%%%%%%%%%%%%%%%%%%%%%%%%%%%%%%%%%%%%%%%%
\end{eqnarray}
\end{widetext}
\end{enumerate}
\end{proposition}

%\subsection{Representation of n-partite quantum states}
%\begin{proposition}
%Let $\rho \in \mathsf{L}(\mathsf{R} \otimes \mathsf{E}_{1} \otimes \cdots \otimes \mathsf{E}_{n-2} \otimes \mathsf{S}_{n})$ be an arbitrary quantum state. Then, there are pure states $\psi \in \mathsf{L}(\mathsf{R} \otimes \mathsf{S}_{1})$, $\varphi_{i} \in \mathsf{L}(\mathsf{F}_{i})$ ($i=1,\cdots,n-2$), unitary channels $\Phi_{i} \colon \mathsf{L}(\mathsf{S}_{i} \otimes \mathsf{F}_{i}) \rightarrow \mathsf{L}(\mathsf{E}_{i} \otimes \mathsf{S}_{i+1})$ ($i=1,\cdots,n-2$), and a channel $\Lambda \colon \mathsf{L}(\mathsf{S}_{n-1}) \rightarrow \mathsf{L}(\mathsf{S}_{n})$ such that
%\begin{align}
%\rho&=(\mathsf{R}\otimes\mathsf{E}_{1}\otimes\cdots\mathsf{E}_{n-2})} \otimes \Lambda)\circ(\mathsf{R}\otimes\mathsf{E}_{1}\otimes\cdots\mathsf{E}_{n-3})} \otimes \Phi_{n-2})  \circ \cdots \circ(\mathsf{R}\otimes\mathsf{E}_{1})} \otimes \Phi_{2} \otimes \mathsf{F}_{3} \otimes \cdots \otimes \mathsf{F}_{n-2})}) \\ & \circ (\mathsf{R})}\otimes \Phi_{1} \otimes \mathsf{F}_{2} \otimes \cdots \otimes \mathsf{F}_{n-2})}) (\psi \otimes \varphi_{1} \otimes \cdots \otimes \varphi_{n-2})
%\end{align}
%\end{proposition}

\begin{widetext}
\subsection{Choi state of a CPTP map as an action of a CP map on the adjoint space}\label{cpslide}

Below we show how the Choi state of a CPTP map is identical to the action of a unital CP on the adjoint space.
\begin{align}
    \mathrm{A}\otimes \mathrm{id}(\Psi^+) 
    %%%%%
    =& \sum_{ijk} A_k \otimes \openone \left| ii \right> \!\left< jj \right| A_k^\dag \otimes \openone \\
    %%%%%
     =& \sum_{ijk}\left(\sum_{mn} a^{(k)}_{nm} \left|n\right>\!\left< m \right| \otimes \openone \right)
     \left| ii \right>\!\left< jj \right| 
     \left( \sum_{rs}  a^{(k)*}_{rs} \left|s\right>\!\left< r \right| \otimes \openone  \right) \notag\\
     %%%%%
    %  =& \sum_{ijk} a^{(k)}_{nm} \left| nm \right>
    %  \!\left< rs \right| a^{(k)*}_{rs} \\
     %%%%%
     =& \sum_{ijk}\left(\openone \otimes \sum_{mn} a^{(k)}_{nm} \left|m\right>\!\left<n \right|  \right)
     \left| ii \right>\!\left< jj \right| 
     \left( \openone \otimes \sum_{rs}  a^{(k)*}_{rs} \left|r\right>\!\left<s \right|  \right) \notag\\
     %%%%%
    % =& \sum_{ijk} \openone \otimes A_k^T \left| ii \right> \!\left< jj \right| A^*_k \otimes \openone \\
    %  %%%%%
     =& \sum_{ijk} \left(\openone \otimes A_k^\dag \left| ii \right> \!\left< jj \right| A_k \otimes \openone \right)^T
     %%%%%
    =: \mathrm{id} \otimes \mathrm{A}^\dag (\Psi^+) 
\end{align}
The adjoint channel will be unital if the $\mathrm{A}$ is trace preserving. Both channels are CP.
\end{widetext}

%%%%%%%%%%%%%%%%%%%%%%%%%%%%%%%%%%%%%%%%%%%%%%%%%%%%%%
%%%%%%%%%%%%%%%%%%%%%%%%%%%%%%%%%%%%%%%%%%%%%%%%%%%%%%
%%%%%%%%%%%%%%%%%%%%%%%%%%%%%%%%%%%%%%%%%%%%%%%%%%%%%%

%\begin{thebibliography}
%\end{thebibliography}

\bibliography{refs}

\end{document}